\documentclass[journal]{IEEEtran}
\usepackage{amsthm,amsmath,amssymb}
\usepackage{mathrsfs}
\usepackage{setspace}
\usepackage{algorithm}
\usepackage{algpseudocode}
\usepackage{booktabs}
\usepackage{graphicx}
\usepackage{float}
\usepackage{subfigure}
\usepackage{caption}

\allowdisplaybreaks[4]

\newtheorem{lemma}{Lemma}
\begin{document}
%
\title{GJRA: A Global Joint Resource Allocation Scheme \\ for UAV Service of PEC in IIoTs}
%
%
%

\author{Jin~Wang,
        Caiyan~Jin,
        Qiang~Tang,
        Naixue~Xiong$^*$
\thanks{Jin Wang, Caiyan Jin and Qiang Tang, School of Computer Science and Communication Engineering, Changsha University of Science and Technology, Changsha 410114, Hunan, China, Email: jinwang@csust.edu.cn, tangqiang@csust.edu.cn.
}
\thanks{Naixue Xiong, Tianjin Key Laboratory of Advanced Networking College of Intelligence and Computing Tianjin University, Tianjin 300000, China.}%
}

\maketitle

\begin{abstract}
The Industrial Internet of Things (IIoT) is an emerging paradigm to make industrial operations more efficient and intelligent by deploying a massive number of wireless devices to industry scenes. However, due to the limited computing capability and batteries, the Industrial Internet of Things Devices (IIoTDs) can't perform the computation-intensive or delay-sensitive tasks well and provide long-term services in practical. To tackle these challenges, we present an effective global joint resource allocation scheme for Unmanned Aerial Vehicle (UAV) service of Pervasive Edge Computing (PEC) in Industrial Internet of Things (IIoTs) and studied a collaborative UAV server-IIoTDs scheme in this paper. In our proposed scheme, the IIoTDs can keep high-efficiency performance even their battery ran out, by deploying the UAV as a PEC server and a mobile power source to provide task offloading and energy harvesting opportunities for IIoTDs. In consideration of the offloading position selection, devices resource allocation and system performance, we aim to minimize the overall service latency of all IIoTDs consisting of task computation latency and offloading latency, by joint optimizing the task offloading decisions, charging resources allocation, connection management, and UAV computation resources allocation. However, the formulated optimization problem is a mixed-integer nonlinear programming (MINLP) problem which is challenging to solve in general. In order to address the problem, we decompose it into multiple convex sub-problems based on block-coordinate descent (BCD) method to obtain the optimal solution. Performance evaluation demonstrates that our scheme outperforms the existing schemes in terms of the overall service latency of IIoTDs.
\end{abstract}

\begin{IEEEkeywords}
Industrial Internet of Things, Mobile edge computing, Wireless power transmission, Unmanned aerial vehicle, Task offloading.
\end{IEEEkeywords}

%
\IEEEpeerreviewmaketitle

\section{Introduction}
%
%
%
%
\IEEEPARstart{T}{he} industrial application of IoT, or IIoT, as a new industrial concept, combines intelligent machines, advanced analysis and machine-human collaboration together, making industrial operation intensely efficient and intelligent[1-2]. With the development of wireless sensor-actuator networks (WSAN), and wireless sensor networks (WSN), more and more IoT devices (IoTDs) are deployed in oil production platforms, underground mines[3], container ports and hydroelectric stations to measure important operational and environmental parameters. IIoT is anticipated to have the capability to transform many industries, including manufacturing, agriculture, engineering industry and energy industry. However, many IIoTDs have limited computing capability and batteries due to their limited size,and is difficult to replace due to the work environment. In order to maintain the quality of service (QoS), it is necessary to assist these kinds of IIoTDs in processing data.

As for the above problems, PEC and Wireless Power Transfer (WPT) are recognized as the feasible solutions. Faced with the data computation, PEC is a emerging computing paradigm with great potential to enhance the performance of devices by task offloading[4], where data can be processed on the edge of the network[5], with the assistance of intelligent devices.
Faced with the energy supplement, WPT is a technology to realize the vision of IIoT[6], which is designed to provide a stable and controllable wireless power[7]. With the Energy Harvesting (EH), the IIoTDs can power themselves by harvesting the wireless signal. Howerver, due to the limited hardware capacity and propagation loss, the radio frequency (RF) signals over long distances lead to poor performance of the WPT and EH systems.

Thanks to their high mobility, flexible deployment and low cost[8-9], UAVs have been widely used in various scenarios (e.g. search and rescue, cargo delivery, surveillance and monitoring, etc.) as moving relays and flying BSs to enlarge network coverage and enhance the communication quality. Distinguished from the fixed location BSs on the ground, UAVs can not only satisfy different quality-of-service (QoS) requirements, but also more likely to establish the line-of-sight (LoS) links with IIoTDs by adjusting their locations flexibly, which can achieve better communication channels and more reliable transmission quality.

In this paper, we present a global joint resource allocation scheme for UAV service of PEC in IIoTDs system, as illustrated in Fig. 1, in which a moving UAV is deployed as a PEC server and a mobile power source. Specifically, we consider the environments where the terrestrial wireless connection between IIoTDs and ground BSs or APs can't be be established, since there is no such wireless infrastructures exist or have been badly damaged. Besides, the IIoTDs can not perform the sensing tasks while data calculation due to the limitation of hardware. Therefore, a UAV deployed to provide task offloading opportunities and energy supply to IIoTDs in such environments, and the collected data from IIoTDs needs to be execute rapidly. We formulate the system process as a optimization problem aiming at minimizing the sum service latency of all IIoTDs consisting of task computation latency and offloading latency, by joint optimizing the task offloading decisions, charging resources allocation, connection management, and UAV computation resources allocation.

\begin{figure}[t]
\centering
\includegraphics[scale=0.15]{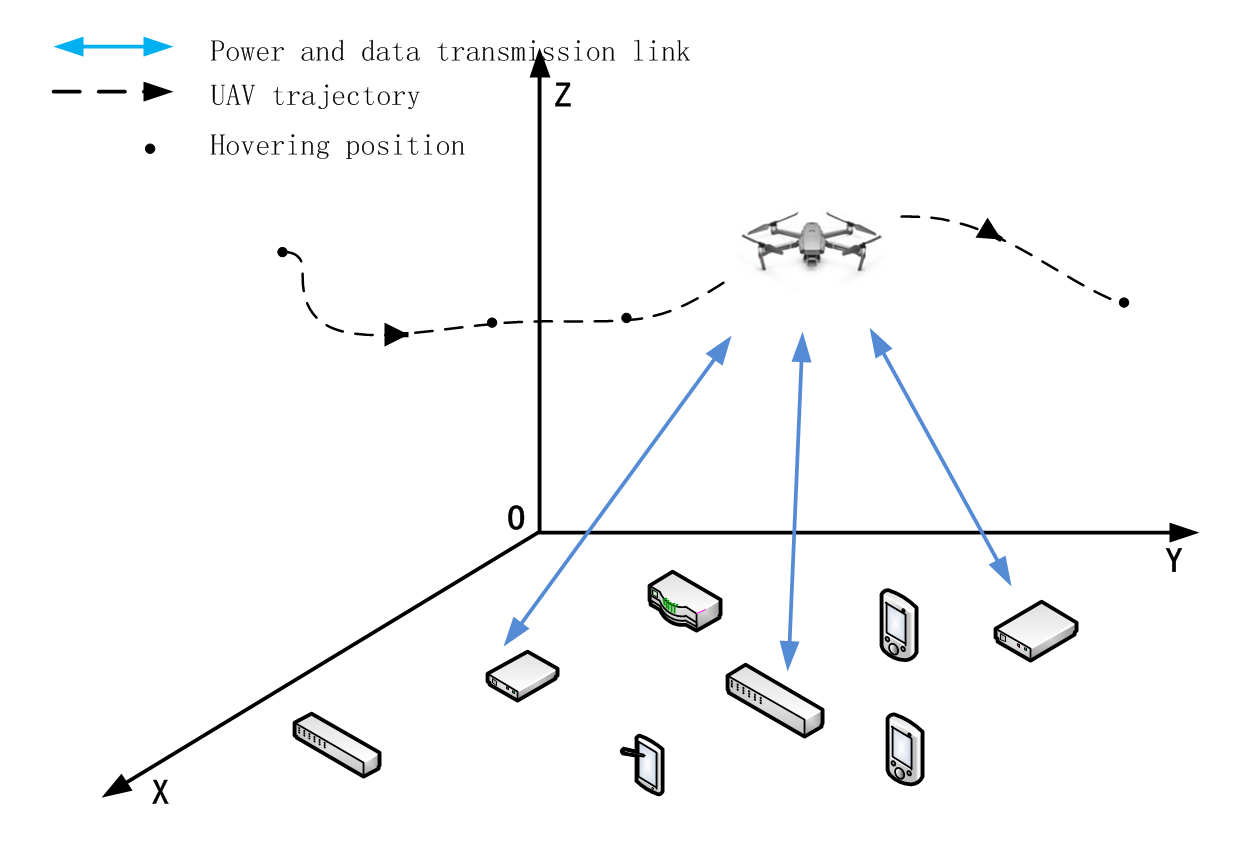}
\centering
\captionsetup{font={small,it}}
\caption{The proposed global joint resource allocation scheme for UAV service of PEC in IIoTDs system}
\label{figl}
\end{figure}
However, the above formulated optimization problem is indeed challenging to tackle. There are two main reasons.On the one hand, there exists correlation among different optimization variables, such as the charging power, the UAV CPU frequency allocation and the variables related to the connection management, making the objective function and constraints non-convex. On the other hand, the variables related to the offloading and connecting decisions are binary, making the problem a mixed integer non-convex optimization problem.

The main contributions of our work are summarized as follows:
1) We propose the collaborative UAV-IIoTDs resource allocation scheme for PEC in IIoTs systems where the UAV is deployed to provide PEC and power-charging services for IIoTDs;
2) Considering the limitation and the QoS requirement of IIoTDs, we formulate the global joint resources allocation as an optimization problem under the proposed system, with the goal of minimizing the sum latency of all IIoTDs;
3) We present an alternating optimization algorithm based on the block-coordinate descent (BCD) method to decouple the optimization variables and develop an heuristic adjusting-approaching algorithm to solve the subproblem relating to the task offloading decisions optimization;
4) To illustrate the performance of the proposed scheme, massive evaluations were conducted. Performance analysis demonstrate that our algorithm can enhance the performance of PEC IIoT systems significantly, compared to several conventional schemes.

The rest of this paper is organized as follows. In Section II, we present the related work. We introduce our system model in Section III. The optimization problem is formulated and solved by the proposed method in Section IV. In Section V, we present our performance analysis. Finally, we conclude the paper and mention the future work in Section VI.

\section{Related Work}
\subsection{Resource Allocation in EC}
Extensive efforts have been dedicated on the resource allocation in edge computing (EC) that aims at optimizing operating cost[10-11], system latency[12], energy consumption[13-14] and network throughput[15].
Wang \emph{et al.} in [10] studied the mobility-agnostic online resource allocation by solving the optimization problem of allocation costs, reconfiguration, service quality and migration under unpredictable resource prices and user movement.
Wang \emph{et al.} in [11] proposed a dynamic optimization scheme for the IoT fog computing system with multiple mobilr devices, aiming at minimizing the system cost by joint optimizaing the radio and computational resources and offloading decisions.
Zhao \emph{et al.} in [12] formulated a cloud-MEC collaborative computation offloading problem through jointly optimizing computation offloading decision and computation resource allocation.
In [13], Zhang \emph{et al.} studied the joint optimization of bits allocation, time slot scheduling, power allocation and UAV trajectory design, aiming at minimizing the total energy consumption.
In [14], Yang et al. investigated joint resource allocation and trajectory design in a MEC network where multiple UAVs are deployed to compute users' offloading tasks,aiming at minimizing the sum energy consumption.
Ning \emph{et al.} in [15] put forward a hybrid computation offloading framework for real-time traffic management aiming at maximizing the sum offloading rate by joint optimizing task distribution, sub-channel assignment and power allocation.

All these studies assume users or devices have sufficient batteries to complete the task transmission and execution. However, it is utmost important to take the limited batteries into consider for enhancing the system endurance. In the view of the above consideration, we propose to deploy a UAV in the PEC IIoTs system to assist the IIoTDs task offloading and wireless charging processes.
\subsection{UAV-assisted WPT}
There are a number of studies on UAV-assisted WPT that aims at trajectory design[16-17], trajectory design based on energy optimization[18-19], communication quality optimization[20] and charging resource allocation[21-22]. Yang \emph{et al.} in [16] proposed a genetic algorithm based successive hover-and-fly scheme to design the optimal UAV trajectory with the objective to maximize the minimal received energy among all users under UAV speed constraints. Ku \emph{et al.} [17] applied Q-learning among reinforcement learning techniques to design UAV trajectory in a WPT system where UAV broadcasts power to energy receivers (ERs) on the ground to solve the fairness problem. Beak \emph{et al.} in [18] deployed UAV as a flying data collector and wireless power source in wireless charging sensor networks (WCSNs). The problem of joint optimization of the UAV hovering location and duration under data collection along with UAV energy consumption constraints,aims at maximizing the minimum energy consumption of sensors after data transmission and energy harvesting.  Xie \emph{et al.} in [20] formulated the system throughput maximization problem under two paradigms of delay-tolerant case and delay-sensitive case by joint optimizing the time slot scheduling, power allocation along with UAV trajectory constrained by a so-called neutrality constraints.
Yin \emph{et al.} in [21] studied the sum of download rate maximization problem in a UAV-assisted cellular network where UAVs are powered by a ground wireless charging station,by joint optimization for user association,resource allocation and station placement.
Chen \emph{et al.} in [22] presented an investigation on the optimal the overall power transmission efficiency considering the UAV's trajectory along with the power of the charging, where the UAV is deployed to collect data reliably form a group of sensors.

To the best of our knowledge, the resource allocation for UAV service of PEC in IIoTs systems has not been well investigated. Therefore, a new model is required in such systems, which is discussed in the following section.

\begin{table}[t]	
\centering
\caption{LIST OF SYMBOLS}\label{table:1}
\resizebox{255pt}{160pt}{\begin{tabular}{c||l}
\toprule
\textbf{Parameter} & \textbf{Description} \\
\midrule
N & Number of IIoTDs \\
M & Number of UAV hovering positions \\
${\cal N}$ & Set of IIoTDs \\
${\cal N}_0$ & Set of IIoTDs in local execution mode \\
${\cal N}_1$ & Set of IIoTDs in task offloading mode \\
${\cal M}$ & Set of UAV hovering positions \\
$w_i$ & Horizontal coordinate of IIoTD \emph{i} \\
$q_j$ & Horizontal coordinate of \emph{j}-th hovering position \\
H & Altitude of the UAV \\
${I_i}$ & Computation task of IIoTD \emph{i} \\
${D_i}$ & Data size of task ${I_i}$ \\
${F_i}$ & Number of CPU cycles of task ${I_i}$ \\
${L_i}$ & Total service latency of task ${I_i}$ \\
$d_i[j]$ & Distance between \emph{j}-th hovering position and IIoTD \emph{i} \\
$a_i[j]$ & Connection status indicator \\
${\Lambda _i}[j]$ & Average pathloss of the IIoTD \emph{i} at \emph{j}-th position \\
${r_i}[j]$ & Channel power gain of the IIoTD \emph{i} at \emph{j}-th position \\
${r_i}[j]$ & Offloading transmission rate of the IIoTD \emph{i} at \emph{j}-th position \\
${p_i}[j]$ & Charging power allocated to the IIoTD \emph{i} at \emph{j}-th position \\
${f_i}^l$ & On-chip computing capability of the IIoTD \emph{i} \\
${f_i}^o[j]$ & Computing resources allocated to the IIoTD \emph{i} at \emph{j}-th position \\
$E_i^{eh}$ & Total harvesting energy of IIoTD \emph{i} \\
$E_i^{l}$ & Total local computation energy consumption of IIoTD \emph{i} \\
$E_i^{tr}$ & Total task offloading energy consumption of IIoTD \emph{i} \\
$T_i^{eh}$ & Total energy harvesting latency of IIoTD \emph{i} \\
$T_i^{l}$ & Task local computation latency of IIoTD \emph{i} \\
$T_i^{tr}$ & Total task offloading transmission latency of IIoTD \emph{i} \\
$T_i^{o}$ & Total task offloading computation latency of IIoTD \emph{i} \\
${\rho _i}$ & Task offloading decision indicator of IIoTD \emph{i} \\
${\eta _0}$ & The energy conservation efficiency \\
${\varphi _i}$ & The effective switched capacitance constant of IIoTD \emph{i} \\
\bottomrule
\end{tabular}}
\end{table}

\section{System Model}
\subsection{Set-Up}
As shown in the Fig.1, we consider a global joint resource allocation scheme for UAV service of PEC in IIoTs system, where a UAV equipped with multiple orthogonal isotropic antennas, is deployed to provide task offloading opportunities and energy supplement to \emph{N} IIoTDs equipped with one single antenna, each of which has an computation-intensive and latency-critical task. The task completion process includes: i) energy transmission and harvesting; ii) task local execution or offloading execution (data migration and assistance computation); iii) result uploading for local execution or beacon post-back for task offloading. We ignore the latency of step iii) due to the small amount of data. We assumed that each IIoTD has an individual computation-intensive and latency-critical task. It is also assumed that the UAV can perform energy transmitting and offloading computing while IIoTDs can perform energy harvesting and local computing or task offloading. For data transmission, in order to avoid interference among IIoTDs, we consider the orthogonal frequency division multiplexing (OFDM) scheme. The main sysbols mentioned in the paper are summarized in \textbf{Table I}.

Without loss of generality, a three-dimensional (3D) Euclidean coordinate is adopted, whose coordinates are measured in meters, and all the devices in the wireless IIoT system are distributed in the first quadrant. We assume that there are a total of \emph{N} IIoTDs randomly distributed in the area and locations of all the IIoTDs are fixed on the ground with zero altitude, with \emph{$w_i$ = ($x_i$,$y_i$)} representing the location of IIoTD \emph{i}, where \emph{i} $\in$ $\mathcal{N}$ and \emph{$\mathcal{N}$ = \{1,2,...,N\}}. Denote \emph{$D_i$} as the amount of transmission data and \emph{$F_i$} as the required processing CPU cycles for task \emph{i}. Thus,we can express the task of IIoTD \emph{i} as:
\begin{equation}
\label{E1}
{I_i}{\rm{ = }}({D_i},{F_i}),\quad \forall i \in {\cal N},
\end{equation}

In addition, we assume that the UAV flies above the area at a fixed altitude \emph{H} and hovers at \emph{M} given locations, with \emph{$q_j$ = ($X_j$,$Y_j$)} representing the location of UAV$'$s \emph{j}-th hovering position, where \emph{j $\in$ $\mathcal{M}$} and \emph{$\mathcal{M}$ = \{1,2,...,M\}}. Therefore, at the hovering position \emph{j}, the distance between UAV and IIoTD \emph{i} is shown as:
\begin{equation}
\label{E2}
{d_i}[j] = \sqrt {{{({X_j} - {x_i})}^2} + {{({Y_j} - {y_i})}^2} + {H^2}},
\end{equation}

Assume each IIoTD can select one and only one UAV hovering position to harvest energy and offload its data, while UAV can serve more than one IIoTD at each hovering position. Therefore, we define binary variables $a_i[j]$ to indicate the connection status between UAV and IIoTDs, where $a_i[j]$ = 1 means the IIoTD \emph{i} chooses the \emph{j}-th UAV hovering position to harvest energy and offload data; otherwise, $a_i[j]$ = 0. It yields the following constraints:
\begin{equation}
\label{E3}
\sum\limits_{j \in {\cal M}}^{} {{a_i}[j] = 1} ,\quad\forall i \in {\cal N},
\end{equation}
\begin{equation}
\label{E4}
{a_i}[j] = \{ 0,1\} ,\quad{\rm{ }}\forall i \in {\cal N},\forall j \in {\cal M},
\end{equation}

\subsection{Channel Model}
In the UAV-enabled wireless powered mobile edge network, we consider the effect of the environment on the occurrence of LoS and an air-to-ground propagation model in suburban environment proposed in [23-25]. In hovering position \emph{j}, the LoS and NLoS pathloss between UAV and IIoTD \emph{i} is given by:
\begin{equation}
\label{E5}
P{L_{LoS,i}}[j] = L_{FS} + 20\log ({d_i}[j]) + {\eta _{LoS}},
\end{equation}
\begin{equation}
\label{E6}
P{L_{NLoS,i}}[j] = L_{FS} + 20\log ({d_i}[j]) + {\eta _{NLoS}},
\end{equation}
where $L_{FS}$ denotes the free space pathloss given by $L_{FS} = 20\log (f) + 20\log (\frac{{4\pi }}{c})$, and  \emph{f} is the system carrier frequency. ${\eta _{LoS}}$ and ${\eta _{NLoS}}$ represent the additional attenuation factors in cases of the LoS and NLoS connections respectively.

The probability of LoS connection is given by:
\begin{equation}
\label{E7}
{P_{LoS,i}}[j] = \frac{1}{{1 + a \cdot \exp ( - b({\theta_i}[j] - a))}},
\end{equation}
where \emph{a} and \emph{b} are constants depending on the environment and ${\theta_i}[j]$ denoted the elevation angle given by ${\theta_i}[j] = {\arctan (\frac{H}{{{d_i}[j]}})}$.

The average pathloss of the IIoTD \emph{i} at \emph{j}-th hovering position is given by:
\begin{equation}
\label{E8}
{\Lambda _i}[j]= {P_{LoS,i}}[j] \cdot P{L_{LoS,i}}[j] + (1 - {P_{LoS,i}}[j]) \cdot P{L_{NLoS,i}}[j],
\end{equation}

We define the \emph{B} as the channel bandwidth and ${p_i}$ as the transmitting power of IIoTD \emph{i}, along with the ${\sigma ^2}$ as the noise power. Then, the transmission rate of IIoTD \emph{i} at \emph{j}-th hovering position is given by[25]:
\begin{equation}
\label{E9}
{r_i}[j] = B{\log _2}(1 + \frac{{{p_i}}}{{{\sigma ^2}{{10}^{{\Lambda _i}[j] /10}}}}),\quad{\rm{ }}\forall i \in {\cal N},\forall j \in {\cal M},
\end{equation}
\subsection{Wireless Energy Harvesting Model}
In the proposed system,the energy consumption of the IIoTDs for local computing and task offloading all comes from the harvested energy. Similar to the [26-27], we applied the linear energy harvesting model in this paper. Thus,the energy harvested by IIoTD \emph{i} at \emph{j}-th the hovering position is given as:
\begin{equation}
\label{E10}
E_i^{eh}[j] = {\eta _0}\sum\limits_{j \in {\cal M}} {{a_i}[j]{p_i}[j]{g_i}[j]T_i^{eh}[j]},
\end{equation}
where ${g_i}[j] = \frac{{{g_0}}}{{{d_i}[j]}}$ is the channel power gain of the IIoTD \emph{i} at \emph{j}-th the hovering position, ${g_0}$ represents the received power at the reference distance ${d_0}$ = 1 m. And ${\eta _0} \in$ (0,1] denotes the energy conservation efficiency, ${p_i}[j]$ denotes the charging power allocated to the IIoTD \emph{i} at \emph{j}-th hovering position and $T_i^{eh}[j]$ denotes the corresponding energy harvesting time.
\subsection{Working Pattern Model}
As mentioned above, each IIoTD can choose computing its task locally, which is the local execution mode, or offloading task to the UAV, which is the task offloading mode. Thus, if IIoTD \emph{i} choose the local execution mode, it will allocate the frequency  $f_{\rm{i}}^l$ for its own task data processing. On the contrary, if IIoTD \emph{i} choose the task offloading mode and offload its task at the \emph{j}-th hovering position, the UAV  will allocate the frequency $f_{\rm{i}}^o[j]$ for the task ${I_i}$ data processing.

In order to distinguishing the two working pattern of IIoTDs preferably, we denote ${{\cal N}_0}$ and ${{\cal N}_1}$ as the set of IIoTDs choosing computing locally and offloading task, respectively. Therefore, ${\cal N} = {{\cal N}_0} \bigcup {{\cal N}_1}$ and ${{\cal N}_0} \bigcap {{\cal N}_1} = \oslash$, where $\oslash$ denotes the null set.

\subsubsection{Local Execution Mode}
For the local execution mode, the computational task of IIoTDs are performed locally. The local execution time is given as:
\begin{equation}
\label{E11}
T_i^l{\rm{ = }}\frac{{{F_i}}}{{f_i^l}},
\end{equation}

The corresponding energy consumption is given as:
\begin{equation}
\label{E12}
E_i^l[j]{\rm{ = }}{\varphi _i}{a_i}[j]{(f_i^l)^v}T_i^l{\rm{ = }}{\varphi _i}{F_i}{a_i}[j]{(f_i^l)^{v{\rm{ - }}1}},
\end{equation}

\noindent where ${{\varphi _i} \ge 0}$ denotes the effective switched capacitance of IIoTD \emph{i} and $v$ denotes the positive constant.

At \emph{j}-th hovering position, the local computing energy consumption of IIoTD \emph{i} should not be more than the total harvesting energy. Thus, one can have:
\begin{equation}
\label{E13}
E_i^l[j] \le E_i^{eh}[j],\quad{\rm{ }}{\rm{ }}\forall i \in {{\cal N}_0},
\end{equation}
\subsubsection{Task Offloading Mode}
For the task offloading mode, IIoTDs will offload their task to the UAV. Based on the channel model mentioned above, the transmission delay and energy consumption for IIoTD \emph{i}$'$s task offloading at the \emph{j}-th hovering position are given as:
\begin{equation}
\label{E14}
T_i^{tr}[j] = \frac{{{I_i}}}{{{r_i}[j]}},
\end{equation}
and
\begin{equation}
\label{E15}
E_i^{tr}[j] = {a_i}[j]{p_i}T_i^{tr}[j] = {a_i}[j]{p_i}\frac{{{I_i}}}{{{r_i}[j]}},
\end{equation}

The task computation time on the UAV is given as:
\begin{equation}
\label{E16}
T_i^o[j] = \frac{{{F_i}}}{{f_i^o[j]}},
\end{equation}

At \emph{j}-th hovering position, the task offloading energy consumption of IIoTD \emph{i} should not be more than the total harvesting energy. Thus, one can have:
\begin{equation}
\label{E17}
E_i^{tr}[j] \le E_i^{eh}[j],\quad{\rm{ }}{\rm{ }}\forall i \in {{\cal N}_1},
\end{equation}
\section{Our Proposed GJRA Scheme}
In order to specify the service delay of IIoTDs, we make following assumptions: (i)IIoTDs cannot execute or offload its task until completing the energy harvesting; (ii)the UAV cannot computing a task until receiving its entire data.
 Therefore, the service latency of IIoTD \emph{i} is given as:
\begin{equation}\label{E18}
{L_i}{\rm{ = }}\left\{ {\begin{array}{*{20}{c}}
{\sum\limits_{j \in {\cal M}} {{a_i}[j](T_i^{eh}[j]{\rm{ + }}T_i^l)} {\rm{                     }},\ {\rm{ }}\forall i \in {{\cal N}_0}}\\
{\sum\limits_{j \in {\cal M}} {{a_i}[j](T_i^{eh}[j] + T_i^{tr}[j] + T_i^o[j])} {\rm{   }},\ {\rm{ }}\forall i \in {{\cal N}_1}}
\end{array}} \right.
\end{equation}

 Assume that the locations of IIoTDs and the UAV's hovering positions are fixed and known[28]. Let ${\bf{A}}  =  \{ {a_i}[j],\forall i \in {\cal N},\forall j \in {\cal M}\}$, ${{\bf{F}}^{\bf{o}}}  =  \{ f_i^o[j],\forall i \in {\cal N},\forall j \in {\cal M}\}$, ${\bf{P}}  =  \{ {p_i}[j],\forall i \in {\cal N},\forall j \in {\cal M}\} $. Our problem becomes to the joint optimization of the task offloading decisions(i.e., ${{\cal N}_0}$ and ${{\cal N}_1}$), the IIoTD connection management(i.e., ${\bf{A}}$), the charging resources allocation(i.e., ${\bf{P}}$) and the UAV computation resources allocation(i.e., ${{\bf{F}}^{\bf{o}}}$), with the goal of minimizing the overall service delay of all IIoTDs. Then, it can be formulated as the following optimization problem:
\begin{subequations}  \label{eq:19}%
\begin{align}
\textbf{P1}:&\mathop {\min }\limits_{{\bf{A}},{\bf{P}},{{\bf{F}}^{\bf{o}}},{{\cal N}_0},{{\cal N}_1}} {\rm{   }}\sum\limits_{i \in {\cal N}} {\sum\limits_{j \in {\cal M}} {{a_i}[j]T_i^{eh}[j]} }  + \sum\limits_{i \in {{\cal N}_0}} {\sum\limits_{j \in {\cal M}} {{a_i}[j]\frac{{{F_i}}}{{f_i^l}}} }, \notag \\
 &\qquad \qquad + \sum\limits_{i \in {{\cal N}_1}} {\sum\limits_{j \in {\cal M}} {{a_i}[j]\left( {\frac{{{F_i}}}{{f_i^o[j]}} + T_i^{tr}[j]} \right)} }              \label{eq:19A} \\
s.t.\quad &{k_i}{F_i}{(f_i^l)^2} \le {\eta _0}{p_i}[j]{g_i}[j]T_i^{eh}[j],{\rm{ }}\forall i \in {{\cal N}_0},{\rm{ }}\forall j \in {\cal M},           \label{eq:19B} \\
&{p_i}T_i^{tr}[j] \le {\eta _0}{p_i}[j]{g_i}[j]T_i^{eh}[j],\ \ \forall i \in {{\cal N}_0},{\rm{ }}\forall j \in {\cal M}, \label{eq:19C}\\
&\sum\limits_{i \in {{\cal N}_1}} {{a_i}[j]f_i^o[j]}  \le f_{\max }^{uav}, \qquad \quad \qquad\forall j \in {\cal M},
\label{eq:19D}\\
&f_i^o[j] \ge 0,\ \ \qquad \qquad \qquad\forall i \in {{\cal N}_1},\forall j \in {\cal M},  \label{eq:19E}\\
&\sum\limits_{i \in {\cal N}} {{a_i}[j]{p_i}[j]}  \le p_{\max }^{uav},\quad \qquad \ \qquad\forall j \in {\cal M},\label{eq:19F}\\\vspace{1ex}
&{p_i}[j] \ge 0,\quad \qquad \qquad \qquad\forall i \in {\cal N},\forall j \in {\cal M},\label{eq:19GH}\\
&\sum\limits_{j \in {\cal M}}^{} {{a_i}[j] = 1} ,\qquad\qquad\qquad\qquad\ \ \forall i \in {\cal N}, \label{eq:19H}\\
&{a_i}[j] = \{ 0,1\} ,\quad\qquad\qquad\forall i \in {\cal N},\forall j \in {\cal M},\label{eq:19I}\\
&{\cal N} = {{\cal N}_0} \cup {{\cal N}_1},\quad{{\cal N}_0} \cap {{\cal N}_1} = \emptyset .\label{eq:19J}%
\end{align}
\end{subequations}
where $f_{\max }^{uav}$ denotes the maximum computing frequency of the UAV while $f_{i,\max }^{ue}$ denotes the maximum computing frequency of the IIoTD \emph{i}, and $p_{\max }^{uav}$ denotes the maximum charing power of the UAV. $(19b)$ and $(19c)$ represent the energy consumption should not be more than the harvesting energy for each IIoTD choosing either local execution mode or task offloading mode, respectively. $(19d)$ means the computation resources allocated to all IIoTDs in task offloading mode cannot exceed the total computation capability of the UAV. $(19e)$ guarantees that the offloading computation resources allocated to each IIoTD is non-negative. Similarly, $(19f)$ means the charing power allocated to all IIoTDs cannot exceed the total wireless power capability of the UAV. $(19g)$ guarantees that the charing power allocated to each IIoTD is nonnegative. $(19g)$ and $(19h)$ represent that all of the IIoTDs can select one and only one UAV hovering position to connect to the UAV. $(19i)$ is the task offloading decision constraint. \textbf{P1} is a MINLP problem, which is NP-hard and difficult to be optimally solved in general.

To solve the formulated problem \textbf{P1}, we obtain the approximate optimal solution for each variable in problem \textbf{P1} by the BCD method. Based on it, we proposed an overall optimization algorithm to get an approximation solution of the formulated problem \textbf{P1}. The details of the proposed algorithm are presented as follows.
\subsection{Task Offloading Decisions Optimization}
In order to efficiently solve \textbf{P1}, a binary variable denoted by ${\rho _i}$ is introduced, where ${\rho _i} \in \{ 0,1\}$ and ${\bf{\rho }} {\rm{ = \{ }}{\rho _i}{\rm{,}}\forall i \in {\cal N}{\rm{\} }}$. ${\rho _i}$ = 0 means that the IIoTD \emph{i} performs local execution mode while ${\rho _i}$ = 1 means that the IIoTD \emph{i} performs task offloading mode. Moreover, the task offloading decision indicator variable ${\rho _i}$ is relaxed as a sharing factor ${\rho _i} \in [0,1]$. Further, we can combine the constraints $(19b)$ to $(19c)$. Thus, \textbf{P1} can be rewritten as follow:
\begin{subequations}  \label{eq:20}
\begin{align}
\textbf{P2}:&\mathop {\min }\limits_{\boldsymbol{\rho },{\bf{A}},{\bf{P}},{{\bf{F}}^{\bf{o}}}} {\rm{   }}\sum\limits_{i \in {\cal N}} {\sum\limits_{j \in {\cal M}} {{a_i}[j]\left\{ \begin{array}{l}
T_i^{eh}[j] + (1 - {\rho _i})\frac{{{F_i}}}{{f_i^l}}\\
 + {\rho _i}\left( {\frac{{{F_i}}}{{f_i^o[j]}} + T_i^{tr}[j]} \right)
\end{array} \right\}} } ,               \label{eq:20A} \\
s.t.\quad &(1 - {\rho _i}){k_i}{F_i}{(f_i^l)^2} + {\rho _i}{p_i}T_i^{tr}[j]\le {\eta _0}{p_i}[j]{g_i}[j]T_i^{eh}[j],\notag\\
&\qquad\qquad\qquad\qquad\qquad\qquad\forall i \in {\cal N},\forall j \in {\cal M}, \label{eq:20B} \\
&\sum\limits_{i \in {\cal N}} {{\rho _i}{a_i}[j]f_i^o[j]}  \le f_{\max }^{uav},\quad\qquad\qquad\forall j \in {\cal M},\label{eq:20C} \\
&(19e),(19f),(19g),(19h),(19i)  \notag
\end{align}
\end{subequations}

Given ${\bf{A}}$, ${\bf{P}}$ and ${{\bf{F}}^{\bf{o}}}$, the subproblem of task offloading decisions optimization can be given as:
\begin{subequations}  \label{eq:21}
\begin{align}
\textbf{P2.1}:&\mathop {\min }\limits_{\boldsymbol{\rho }} {\rm{   }}\sum\limits_{i \in {\cal N}} {\sum\limits_{j \in {\cal M}} {{a_i}[j]\left\{
 (1 - {\rho _i})\frac{{{F_i}}}{{f_i^l}} + {\rho _i}\left( {\frac{{{F_i}}}{{f_i^o[j]}} + T_i^{tr}[j]} \right)
 \right\}} } ,               \label{eq:21A} \\
s.t.\quad &(1 - {\rho _i}){k_i}{F_i}{(f_i^l)^2} + {\rho _i}{p_i}T_i^{tr}[j]\le {\eta _0}{p_i}[j]{g_i}[j]T_i^{eh}[j],\notag\\
&\qquad\qquad\qquad\qquad\qquad\forall i \in {\cal N},\forall j \in {\cal M}, \label{eq:21B} \\
&\sum\limits_{i \in {\cal N}} {{\rho _i}{a_i}[j]f_i^o[j]}  \le f_{\max }^{uav},\qquad\forall j \in {\cal M},\label{eq:21C}
\end{align}
\end{subequations}

To achieve the goal of minimizing the overall service latency, the IIoTD will make its offloading decision based on the tradeoff between the achievable minimum task completion time and the necessary resources consumption. Hence, for any given ${\bf{A}}$, ${{\bf{F}}^{\bf{o}}}$ and ${\bf{P}}$, the IIoTD offloading decision scheme depends on the achievable minimum sum service latency.

In this case, while the IIoTD \emph{i} performs local execution mode, the minimum service latency of task \emph{i} can be obtained by constraint $(21b)$ as $\sum\limits_{j \in M} {{a_i}[j]\frac{{k_i}{F_i}{(f_i^l)^2}}{{\eta _0}{p_i}[j]{g_i}[j]}+ T_i^l}$. And similarly, while the IIoTD \emph{i} performs task offloading mode, the minimum service latency of task \emph{i} can be expressed as $\sum\limits_{j \in M} {{a_i}[j] \left( \frac{{p_i}T_i^{tr}[j]}{{\eta _0}{p_i}[j]{g_i}[j]}+ T_i^{tr}[j] + T_i^o[j] \right) }$ by constraint $(21b)$.

Therefore, the partial optimal IIoTD task offloading decision scheme constrained by $(21b)$ can be obtained by

\begin{equation} \label{eq:22}
{\rho _i}^{p\_opt}{\rm{ = }}\left\{ {\begin{array}{*{20}{c}}
{1 \ \ \ \quad   if \ {h_{i,1}} \ge {{h_{i,2}}}}\\
{{0  \qquad \quad     otherwise;}}
\end{array}} \right.
\end{equation}
\begin{equation} \label{eq:23}
h_{i,1} = \sum\limits_{j \in M} {{a_i}[j]\frac{{k_i}{F_i}{(f_i^l)^2}}{{\eta _0}{p_i}[j]{g_i}[j]}+ \frac{{{F_i}}}{{f_i^l}}} ,
\end{equation}
\begin{equation} \label{eq:24}
h_{i,2} = \sum\limits_{j \in M} {{a_i}[j] \left( \frac{{p_i}T_i^{tr}[j]}{{\eta _0}{p_i}[j]{g_i}[j]} + \frac{{F_i}}{{f_i^o[j]}} + T_i^{tr}[j] \right) } ,
\end{equation}
where $h_{i,1}$ is denoted as the service latency of task \emph{i} in local execution mode and $h_{i,2}$ is denoted as the service latency of task \emph{i} in task offloading mode.

Considering the existence of constraint $(21c)$, we developed an heuristic method to find out the optimal solution ${\rho _i}^{opt}$ of problem \textbf{P2.1} by continuously adjusted based on the partial optimal scheme obtained above. The proposed heuristic adjusting-approaching method is specified as \textbf{Algorithm 1}.
\begin{algorithm}[t]
  \caption{Heuristic Adjusting-Approaching Method}
  \label{alg1}
  \begin{algorithmic}[1]
    \Require
        Fixed $\boldsymbol{A, F^o, P}$,the partial optimal scheme $\boldsymbol{\rho^{p\_opt}}$.
    \Ensure
        The optimal solution $\boldsymbol{\rho^{opt}}$
    \State Initialize the current solution $\boldsymbol{\rho^{current}}$ = $\boldsymbol{\rho^{p\_opt}}$;
    \For {$j = 1 \to M$}
    \If {(21c) is not satisfied for $j$-th position according to $\boldsymbol{\rho^{current}}$}
    \Repeat
    \State Obtain $\boldsymbol{H_1} \gets \{h_{i,1} | \forall i \in {\cal N} \}$, $\boldsymbol{H_2} \gets \{h_{i,2} | \forall i \in {\cal N} \}$
    \State Obtain $\boldsymbol{\triangle H_{1-2}} \gets \{h_{i,1} - h_{i,2} | \forall i \in {\cal N}, h_{i,1} > h_{i,2}\}$
    \State Get $k \gets \arg\min$$\boldsymbol{\triangle H_{1-2}}$
    \State Update $\rho_k \gets 0$
    \Until (21c) is satisfied for $j$-th position
    \State Update $\boldsymbol{\rho^{current}}$
    \EndIf
    \EndFor
    \State $\boldsymbol{\rho^{opt}}$ = $\boldsymbol{\rho^{current}}$ \\
    \Return $\boldsymbol{\rho^{opt}}$
  \end{algorithmic}
\end{algorithm}
\subsection{UAV Computing Resource Allocation Optimization}
Given ${\bf{A}}$, ${\bf{P}} $, $\boldsymbol{\rho }$. The sub-problem with regard to ${{\bf{F}}^{\bf{o}}}$ is:
\begin{subequations}  \label{eq:29}
\begin{align}
\textbf{P2.2}:&\mathop{\min }\limits_{{\bf{F}}^{\bf{o}}} \sum\limits_{i \in {\cal N}} {\sum\limits_{j \in {\cal M}} {{a_i}[j]{\rho _i}\frac{{{F_i}}}{{f_i^o[j]}}} } ,                \tag{23} \\
s.t.\quad &(20c), (19e)  \notag
\end{align}
\end{subequations}

Problem \textbf{P2.2} is a convex problem. Thus, \textbf{P2.2} can be solved by the convex optimization technique such as the interior-point method[29]. To gain more insights on the structure of the optimal solution, we leverage the Lagrange method to obtain a well-structured solution. The Lagrange multipliers associated with the constraints in $(20c)$ is given as $\boldsymbol{\mu} = {\{ {\mu  _j}\geq0\} _{j \in {\cal M}}}$.  The partial Lagrangian function of \textbf{P2.2} is
\begin{equation}\label{eq:30}
\begin{split}
{\cal L}(\boldsymbol{F^o, \mu}) = & \sum\limits_{i \in {\cal N}} {\sum\limits_{j \in {\cal M}} {{a_i}[j]{\rho _i}\frac{{{F_i}}}{{f_i^o[j]}}} } \\
& + \sum\limits_{j \in {\cal M}} {{\mu _j}\left( {\sum\limits_{i \in {\cal N}} {{\rho _i}{a_i}[j]f_i^o[j]}  - f_{\max }^{uav}} \right)},
\end{split}
\end{equation}

The dual function of \textbf{P2.2} is given as
\begin{equation}\label{eq:31}
\begin{split}
g(\boldsymbol{\mu}) = & \mathop {\min }\limits_{\bf{F^o}} {\cal L}(\boldsymbol{F^o, \mu}) , \\
s.t. \quad &(19e)
\end{split}
\end{equation}

Then the dual problem of \textbf{P2.2} is given as
\begin{equation}\label{eq:25}
\begin{split}
&\mathop {\max }  \limits_{\boldsymbol{\mu}} \quad g(\boldsymbol{\mu}) = \mathop {\min }\limits_{\bf{F^o}} {\cal L}(\boldsymbol{F^o, \mu}),\\
&\quad s.t.\quad {\mu _j} \ge 0,\forall j \in {\cal M},
\end{split}
\end{equation}

Since the convex problem \textbf{P2.2} satisfies the Slater's condition, strong duality holds between problem \textbf{P2.2} and problem (28). Therefore, one can solve problem \textbf{P2.2} by equivalently solving its dual problem (28).
\subsubsection{Derivation of the Dual Function $g(\boldsymbol{\mu} )$}
Given any $\boldsymbol{\mu}$, we can obtain $g(\boldsymbol{\mu} )$ by solving problem (27). Notice that problem (27) can be decomposed into the following ${\rm{N}} \times {\rm{M}}$ subproblems:
\begin{equation}\label{eq:33}
\begin{split}
\mathop {\min }\limits_{\bf{F^o}} \quad & {a_i}[j]{\rho _i}\frac{{{F_i}}}{{f_i^o[j]}}{\rm{ + }}{\mu _j}{\rho _i}{a_i}[j]f_i^o[j], \\
s.t. \quad &(19e)
\end{split}
\end{equation}
According to to monotonicity of objective function, the optimal solution of problem (27) is given as
\begin{equation}\label{eq:34}
\vspace{1ex}
{\left( {f_i^o[j]} \right)^{opt}} = \left\{ {\begin{array}{*{20}{c}}
{\sqrt {\frac{{{F_i}}}{{{\mu _j}}},} }&{{\mu _j} > 0,}\\
{f_{\max }^{uav},}&{{\mu _j} = 0.}
\end{array}} \right.
\end{equation}
\subsubsection{Obtaining $\boldsymbol{\mu ^{opt}}$ to Maximize $g(\boldsymbol{\mu})$}
Solving dual problem (28) means obtaining $\boldsymbol{\mu ^{opt}}$ in the defined domain to maximize $g(\boldsymbol{\mu} )$. Putting eq.29 into problem (28), thus we can obtain:
\begin{subequations}  \label{eq:35}
\begin{align}
\mathop {\max }\limits_{\boldsymbol{\mu}} & {\rm{    }}\sum\limits_{j \in {\cal M}} {\left[ {\sum\limits_{i \in {\cal N}} {2{a_i}[j]{\rho _i}\sqrt {{F_i}{\mu _j}} }  - {\mu _j}f_{\max }^{uav}} \right]},                   \label{eq:35A} \\
s.t.\quad&{\rm{ }}{\mu _j} > 0,\quad \forall j \in {\cal M},  \label{eq:35B}
\end{align}
\end{subequations}

Notice that problem (31) can be decomposed into the following M subproblems, one can have:

\begin{equation}\label{eq:36}
\begin{split}
\mathop {\max }\limits_{\boldsymbol{\mu}} \quad & {\rm{    }} { {\sum\limits_{i \in {\cal N}} {2{a_i}[j]{\rho _i}\sqrt {{F_i}{\mu _j}} }  - {\mu _j}f_{\max }^{uav}} },  \\
s.t.\quad&{\rm{ }}{\mu _j} > 0
\end{split}
\end{equation}

According to the monotonicity of the objective function, one can have:
\begin{equation}\label{eq:37}
{\mu _j}^{\rm{opt}}{\rm{ = }}{\left( {\frac{{\sum\limits_{i \in {\cal N}} {{a_i}[j]{\rho _i}\sqrt {{F_i}} } }}{{f_{\max }^{uav}}}} \right)^2} = \frac{{\sum\limits_{i \in {\cal N}} {{a_i}[j]{\rho _i}Fi} }}{{f{{_{\max }^{uav}}^2}}},
\end{equation}

Therefore, the optimal solution to ${\left( {f_i^o[j]} \right)^{opt}}$ can be obtained by
\begin{equation}\label{eq:38}
\begin{split}
{\left( {f_i^o[j]} \right)^{opt}} &= \mathop {\arg \max }\limits_{F^o,\mu } g({\left( {f_i^o[j]} \right)^{opt}},{\mu _j}^{\rm{opt}}) \\
& = \left\{ {\begin{array}{*{20}{c}}
{f_{\max }^{uav}\sqrt {\frac{{{F_i}}}{{\sum\limits_{i \in N} {{a_i}[j]{\rho _i}Fi} }}} }&{{\mu _j} > 0,}\\
{f_{\max }^{uav},}&{{\mu _j} = 0.}
\end{array}} \right. \\
\end{split}
\end{equation}
\subsection{UAV Charging Power Optimization}
With given ${\bf{A}}$, ${{\bf{F}}^{\bf{o}}}$, $\boldsymbol{\rho }$, the sub-problem on optimizing ${\bf{P}} $ is:\begin{subequations}  \label{eq:25}
\begin{align}
\textbf{P2.3}:&\mathop {\min }\limits_{{\bf{P}}} {\rm{   }}\sum\limits_{i \in {\cal N}} {\sum\limits_{j \in {\cal M}} {{a_i}[j]T_i^{eh}[j]} },                \label{eq:25A} \\
s.t.\quad &(1 - {\rho _i}){k_i}{F_i}{(f_i^l)^2} + {\rho _i}{p_i}T_i^{tr}[j]\le {\eta _0}{p_i}[j]{g_i}[j]T_i^{eh}[j],\notag\\
&\qquad\qquad\qquad\qquad\qquad\qquad\forall i \in {\cal N},\forall j \in {\cal M},          \label{eq:25B} \\
&(19f),(19g) \notag
\end{align}
\end{subequations}
\begin{lemma} \label{lemma1}
For problem \textbf{P2.3}, the equal sign always holds for \eqref{eq:25B}.
\end{lemma}
\begin{proof}
\renewcommand{\qedsymbol}

As mentioned above, the service latency for any IIoTD \emph{i} consists of two parts: 1) energy harvesting time $T_i^{eh}$; and 2) task computation time, the local computation time $T_i^l$ or the sum of transmission latency and offloading computation latency $(T_i^{tr}+T_i^o)$.

To achieve the goal of minimizing the overall service latency of all IIoTDs in the system is to minimize the service latency of every IIoTD, which is to minimize the both two parts time consumption mentioned above for every IIoTD. In other words, the optimal energy harvesting time for IIoTD \emph{i}, $T_i^{eh}$, is its lower bound, which is characterized by constraint \eqref{eq:25B}. This thus proves the lemma.
\rightline{ $\blacksquare$}
\end{proof}
Thus, we can obtain:
\begin{subequations}  \label{eq:26}
\begin{align}
\textbf{P2.4}:&\mathop {\min }\limits_{{\bf{P}}} {\rm{   }}\sum\limits_{i \in {\cal N}} {\sum\limits_{j \in {\cal M}} {{a_i}[j]T_i^{eh}[j]} } ,               \label{eq:26A} \\
s.t.\quad&((1 - {\rho _i}){k_i}{F_i}{(f_i^l)^2} = {\eta _0}{p_i}[j]{g_i}[j]T_i^{eh}[j],\notag \\
&\qquad\qquad\qquad\qquad\qquad\qquad\forall i \in {\cal N},\forall j \in {\cal M},          \label{eq:26B} \\
&(19f),(19g) \notag
\end{align}
\end{subequations}

Since the optimal task offloading decisions and ${{\bf{F}}^{\bf{o}}}$ have been obtained above, as well as ${\bf{A}}$ is pre-defined, we can rewrite the \textbf{P2.4} as :

\begin{subequations}  \label{eq:27}
\begin{align}
\textbf{P2.5}:&\mathop{\min }\limits_{{\bf{P}}} {\rm{   }}\sum\limits_{i \in {\cal N}} {\sum\limits_{j \in {\cal M}} {{a_i}[j]\left\{ {\frac{{(1 - {\rho _i}){k_i}{F_i}{{(f_i^l)}^2}}}{{{\eta _0}{p_i}[j]{g_i}[j]}} + \frac{{{\rho _i}{p_i}T_i^{{\rm{tr}}}{\rm{[j]}}}}{{{\eta _0}{p_i}[j]{g_i}[j]}}} \right\}} } ,                 \tag{37} \\
&s.t.\quad(19f),(19g) \notag
\end{align}
\end{subequations}
\begin{figure*}[!t]
\begin{align*} \label{eq:55}
{\cal L}(\boldsymbol{A ,\beta ,\gamma }) = & \sum\limits_{i \in {\cal N}} {\sum\limits_{j \in {\cal M}} {{a_i}[j]\left[ {T_i^{eh}[j] + (1 - {\rho _i})\frac{{{F_i}}}{{f_i^l}} + {\rho _i}\left( {\frac{{{F_i}}}{{f_i^o[j]}} + T_i^{{\rm{tr}}}{\rm{[j]}}} \right)} \right]} } \\
 & + \sum\limits_{j \in {\cal M}}^{} {{\beta _j}(\sum\limits_{i \in {\cal N}} {{\rho _i}{a_i}[j]f_i^o[j]}  - f_{\max }^{uav})} {\rm{ + }}\sum\limits_{j \in {\cal M}}^{} {{\gamma _j}(\sum\limits_{i \in {\cal N}} {{a_i}[j]{p_i}[j]}  - p_{\max }^{uav})}
\tag{50} 
\end{align*}
\hrulefill
\end{figure*}

It can easily proved that problem \textbf{P2.5} is convex. Therefore, as solving problem \textbf{P2.2}, we can leverage the Lagrange method to solve this problem similarly. The Lagrange multipliers associated with the constraints in $(19f)$ is given as $\boldsymbol{\lambda } = {\{ {\lambda _j}\geq0\} _{j \in {\cal M}}}$ and the partial Lagrangian function of \textbf{P2.5} is
\begin{equation}\label{eq:37}
\begin{split}
{\cal L}(\boldsymbol{P,\lambda }) = & \sum\limits_{i \in {\cal N}} {\sum\limits_{j \in {\cal M}} {{a_i}[j]\frac{{(1 - {\rho _i}){k_i}{F_i}{{(f_i^l)}^2} + {\rho _i}{p_i}T_i^{{\rm{tr}}}{\rm{[j]}}}}{{{\eta _0}{p_i}[j]{g_i}[j]}}} } \\
& + \sum\limits_{j \in {\cal M}} {{\lambda _j}\left( {\sum\limits_{i \in {\cal N}} {{a_i}[j]{p_i}[j]}  - p_{\max }^{uav}} \right)},
\end{split}
\end{equation}

Then the dual function of \textbf{P2.5} is given as
\begin{equation}\label{eq:38}
\begin{split}
g(\boldsymbol{\lambda }) = &\mathop {\min }\limits_{\boldsymbol{P}} {\cal L}(\boldsymbol{P,\lambda }), \\
s.t.\quad &(19g)
\end{split}
\end{equation}

Thus, the dual problem of \textbf{P2.5} is
\begin{equation}\label{eq:39}
\begin{split}
&\mathop {\max }\limits_{\boldsymbol{\lambda}}  \quad g(\boldsymbol{\lambda }), \\
s.t.\quad &{\rm{ }}{\lambda _j} \ge 0,\forall j \in {\cal M},
\end{split}
\end{equation}

Strong duality holds between problem \textbf{P2.5} and problem (40) since problem \textbf{P2.5} is convex and it also satisfies the Slater's condition. Therefore, one can solve problem \textbf{P2.5} by equivalently solving its dual problem (40).
\subsubsection{Derivation of the Dual Function $g(\boldsymbol{\lambda })$}
Given any $\boldsymbol{\lambda}$, we can obtain $g(\boldsymbol{\lambda })$ by solving problem (39). Notice that problem (39) can be decomposed into the following ${N \times M}$ subproblems.
\begin{equation}\label{eq:40}
\begin{split}
\mathop {\min }\limits_P {\rm{     }}{a_i}[j]\frac{{(1 - {\rho _i}){k_i}{F_i}{{(f_i^l)}^2} + {\rho _i}{p_i}T_i^{{\rm{tr}}}{\rm{[j]}}}}{{{\eta _0}{g_i}[j]}}\frac{1}{{{p_i}[j]}}{\rm{ + }}{\lambda _j}{a_i}[j]{p_i}[j],
\end{split}
\end{equation}
Let ${\cal F}({p_i}[j]) = {a_i}[j]\frac{{(1 - {\rho _i}){k_i}{F_i}{{(f_i^l)}^2} + {\rho _i}{p_i}T_i^{{\rm{tr}}}{\rm{[j]}}}}{{{\eta _0}{g_i}[j]}}\frac{1}{{{p_i}[j]}}{\rm{ + }}{\lambda _j}{a_i}[j]{p_i}[j]$.

\begin{equation}\label{eq:41}
\frac{{\partial {\cal F}}}{{\partial {p_i}[j]}} = {\lambda _j}{a_i}[j] - {a_i}[j]{A_i}[j]\frac{1}{{{{\left( {{p_i}[j]} \right)}^2}}},
\end{equation}
where ${{\rm{A}}_i}{\rm{[j] = }}\frac{{(1 - {\rho _i}){k_i}{F_i}{{(f_i^l)}^2} + {\rho _i}{p_i}T_i^{{\rm{tr}}}{\rm{[j]}}}}{{{\eta _0}{g_i}[j]}}$.

Let $\frac{{\partial {\cal F}}}{{\partial {p_i}[j]}} = 0$, one can have
\begin{equation}\label{eq:42}
{\left( {p_i[j]} \right)^{opt}} = \left\{ {\begin{array}{*{20}{c}}
{\sqrt {\frac{{{A_i}[j]}}{{{\lambda _j}}}}, }&{{\lambda _j} > 0,}\\
\vspace{1ex}
{p_{\max }^{uav},}&{{\lambda _j} = 0.}  
\end{array}} \right.
\end{equation}
\subsubsection{Obtaining $\boldsymbol{\lambda ^{opt}}$ to Maximize $g(\boldsymbol{\lambda })$}
Solving dual problem (40) means obtaining $\boldsymbol{\lambda ^{opt}}$ in the defined domain to maximize $g(\boldsymbol{\lambda })$. Putting (43) into problem (40), thus we can obtain:
\begin{subequations}  \label{eq:43}
\begin{spacing}{1.2}
\begin{align}
\mathop {\max }\limits_\lambda  {\rm{    }}\sum\limits_{j \in {\cal M}} &{\left[ {\sum\limits_{i \in {\cal N}} {2{a_i}[j]\sqrt {{A_i}[j]{\lambda _j}} }  - {\lambda _j}p_{\max }^{uav}} \right]},               \\
s.t.\quad &{\rm{ }}{\lambda _j} > 0,\quad \forall j \in {\cal M},
\end{align}
\end{spacing}
\end{subequations}

Notice that problem (43) can be decomposed into the following $M$ subproblems.
\begin{equation}\label{eq:44}
\begin{split}
\mathop {\max }\limits_\lambda & {\sum\limits_{i \in {\cal N}} {2{a_i}[j]\sqrt {{A_i}[j]{\lambda _j}} }  - {\lambda _j}p_{\max }^{uav}}  \\
s.t.\quad &{\rm{ }}{\lambda _j} > 0
\end{split}
\end{equation}

According to the monotonicity of objective function, one can have
\begin{equation}\label{eq:45}
{\lambda _j}^{\rm{opt}}{\rm{ = }}{\left( {\frac{{\sum\limits_{i \in {\cal N}} {{a_i}[j]\sqrt {{A_i}[j]} } }}{{p_{\max }^{uav}}}} \right)^2} = \frac{{\sum\limits_{i \in {\cal N}} {{a_i}[j]{A_i}[j]} }}{{p{{_{\max }^{uav}}^2}}}
\end{equation}

Therefore, the optimal solution to ${\left( {{p_i}[j]} \right)^{opt}}$ can be obtained by
\begin{equation}\label{eq:44}
\begin{split}
{\left( {{p_i}[j]} \right)^{opt}} &= \mathop {\arg \max }\limits_{P,\lambda } g({\left( {{p_i}[j]} \right)^{opt}},{\lambda _j}^{\rm{opt}}) \\
& = \left\{ {\begin{array}{*{20}{c}}
{p_{\max }^{uav}\sqrt {\frac{{{A_i}[j]}}{{\sum\limits_{i \in N} {{a_i}[j]{A_i}[j]} }}} }&{{\lambda _j} > 0,}\\
{p_{\max }^{uav},}&{{\lambda _j} = 0.}
\end{array}} \right. \\
\end{split}
\end{equation}

\subsection{IIoTD Connection Management Optimization}
With obtained ${\bf{P}} $ , ${{\bf{F}}^{\bf{o}}}$ and task offloading decision, the sub-problem on optimizing IIoTD connection management ${\bf{A}}$ can be formulated as:
\begin{subequations}  \label{eq:31}
\begin{spacing}{1.3}
\begin{align}
\textbf{P2.6}:&\mathop{\min }\limits_{{\bf{A}}} {\rm{   }}\sum\limits_{i \in {\cal N}} {\sum\limits_{j \in {\cal M}} {{a_i}[j]\left\{ \begin{array}{l}
T_i^{eh}[j] + (1 - {\rho _i})\frac{{{F_i}}}{{f_i^l}}\\
 + {\rho _i}\left( {\frac{{{F_i}}}{{f_i^o[j]}} + T_i^{tr}[j]} \right)
\end{array} \right\}} } \tag{48} \\
s.t.\quad&(20c) ,(19f), (19h), (19i) \notag
\end{align}
\end{spacing}
\end{subequations}
To efficiently solve problem \textbf{P2.6}, the variable ${a_i}[j]$ is relaxed as a sharing factor ${a_i}[j]\in [0,1]$. Thus, \textbf{P2.6} can be transformed as follow:
\begin{subequations}  \label{eq:32}
\begin{align}
\textbf{P2.7}:&\mathop{\min }\limits_{{\bf{A}}} {\rm{   }}\sum\limits_{i \in {\cal N}} {\sum\limits_{j \in {\cal M}} {{a_i}[j]\left\{ \begin{array}{l}
T_i^{eh}[j] + (1 - {\rho _i})\frac{{{F_i}}}{{f_i^l}}\\
 + {\rho _i}\left( {\frac{{{F_i}}}{{f_i^o[j]}} + T_i^{tr}[j]} \right)
\end{array} \right\}} }        ,         \label{eq:32A} \\
s.t.\quad &0 \le {a_i}[j] \le 1,{\rm{ }}\qquad\qquad\qquad\forall i \in {\cal N},\forall j \in {\cal M},  \label{eq:32B} \\
&(20c) ,(19f), (19h)  \notag
\end{align}
\end{subequations}

Obviously, problem \textbf{P2.7} is a convex problem with respect to ${\bf{A}}$, which can be effectively solved via Lagrange method.
Denoting $\boldsymbol{\beta}$  = ${\{ {\beta _j} \ge 0\} _{j \in {\cal M}}}$ and $\boldsymbol{\gamma } = {\{ {\gamma _j} \ge 0\} _{j \in {\cal M}}}$ as the Lagrange multiplier vectors associated with constraints $(19f)$ and $(20b)$, respectively.  Then, the Lagrangian of problem \textbf{P2.7} can be given by (50) at the top of the previous page and the Lagrange dual function of problem \textbf{P2.7} can be presented as:
\begin{equation}\label{eq:56}
\begin{split}
g(\boldsymbol{\beta ,\gamma }) = & \mathop {\min }\limits_{\boldsymbol{A}} {\cal L}(\boldsymbol{A ,\beta ,\gamma }), \\
s.t.& (19h),(49b)
\end{split}
\tag{51}
\end{equation}

Then, the corresponding dual problem is given as:
\begin{equation}\label{eq:57}
\begin{split}
&\mathop {\max }\limits_{\boldsymbol{\beta ,\gamma }} \qquad g(\boldsymbol{\beta ,\gamma }), \\
&s.t.{\beta _j} \ge 0,\ {\gamma _j} \ge 0,\quad \forall j \in {\cal M},
\end{split}
\tag{52}
\end{equation}

\begin{algorithm}[t]
  \caption{Gradient Descent on Lagrange Dual Based Algorithm for IIoTDs Connection Management}
  \label{alg1}
  \begin{algorithmic}[1]
    \Require
        Fixed $\boldsymbol{\rho, F^o, P}$,
        the tolerance of accuracy $\varepsilon_{1}$,the maximum iteration number $K_{max}$.
    \Ensure
        The optimal solution $\boldsymbol{A^{opt}}$
    \State Initialize lagrange multipliers $\{{\beta_{j}}\}_{j\in M},\{{\gamma_{j}}\}_{j\in M}$ and the current iteration number $k = 0$.
    \Repeat
    \State Set $k = k + 1$.
    \State Obtain the optimal IIoTDs connection management $\boldsymbol{A^{opt}}$ according to (53) and (54).
    \State Update lagrange multipliers $\{{\beta_{j}}\}_{j\in M}, \{{\gamma_{j}}\}_{j\in M}$ based on (55) and (56).
    \Until{The difference between consecutive values of the objective function (49a) is under $\varepsilon_{1}$ or $k \geq K_{max}$.}\\
    \Return $\boldsymbol{A^{opt}}$
  \end{algorithmic}
\end{algorithm}

To minimize the objective function in (50) which is a linear combination of ${a_i}[j]$, we can let the connection coefficient corresponding to the UAV with the smallest ${h_i}[j]$ be 1 for any \emph{i}. Therefore, the optimal solution of \textbf{P2.7} is given as:
\begin{equation}\label{eq:33}
{a_i}{[j]^{opt}} = \left\{ {\begin{array}{*{20}{c}}
{1,\qquad if\quad j = \mathop {\arg \max }\limits_{j \in {\cal M}} {h_i}[j];}\\
{0,\qquad\qquad\qquad\quad otherwise.}
\end{array}} \right.
\tag{53}
\end{equation}
where
\begin{equation}\label{eq:34}
\begin{array}{l}
{h_i}[j] =  T_i^{eh}[j] + (1 - {\rho _i})\frac{{{F_i}}}{{f_i^l}} + {\rho _i}\left( {\frac{{{F_i}}}{{f_i^o[j]}} + T_i^{tr}[j]} \right)\\
\hspace{0.5cm}
 + {\beta _j}{\rho _i}f_i^o[j] + {\gamma _j}{p_i}[j],\qquad\forall i \in {\cal N},\forall j \in {\cal M},
\end{array}
\tag{54}
\end{equation}

The values of ${\{ {\beta _j}\} _{j \in {\cal M}}}$ and ${\{ {\gamma _j}\} _{j \in {\cal M}}}$ can be determined by the sub-gradient method[31]. The updating procedure can be given by:\\
\begin{equation}\label{eq:35}
{\beta _j}^{(l + 1)} = {\left[ {{\beta _j}^{(l)} - \omega \left( {\sum\limits_{i \in N} {{\rho _i}{a_i}[j]f_i^o[j]}  - f_{\max }^{uav}} \right)} \right]^ + },
\tag{55}
\end{equation}
\begin{equation}\label{eq:36}
{\gamma _j}^{(l + 1)} = {\left[ {{\gamma _j}^{(l)} - \omega \left( {\sum\limits_{i \in N} {{a_i}[j]{p_i}[j]}  - p_{\max }^{uav}} \right)} \right]^ + },\quad\
\tag{56}
\end{equation}
where ${\left[ x \right]^ + } = \max \{ x,0\} $, and $\omega  > 0$ is a dynamically step-size sequence chosen by the self-adaptive scheme of [30].

The optimal solution of problem \textbf{P2.6} can be obtained via the gradient descent on the Lagrange dual method with zero duality gap, by iteratively optimizing ${a_i}[j]$ in (52) and (53) and updating ${\{ {\beta _j}\} _{j \in {\cal M}}}$ and ${\{ {\gamma _j}\} _{j \in {\cal M}}}$ according (54) and (55).

The gradient descent on the Lagrange dual based algorithm for solving problem \textbf{P2.6} with fixed $\boldsymbol{\rho, F^o, P}$ is given by \textbf{Algorithm 2}. Moreover, the overall iterative algorithm GJRA is given in \textbf{Algorithm 3}.
\begin{algorithm}[t]
  \caption{Overall Algorithm GJAR}
  \label{alg2}
  \begin{algorithmic}[1]
    \Require
        $\boldsymbol{A^0, \rho ^0, (F^o)^0, P^0}$,
        the tolerances of accuracy $\varepsilon_{1}$ and $\varepsilon_{2}$,
        the maximum iteration number $K_{max}$ and $R_{max}$.
    \Ensure
         The optimal Users association $\boldsymbol{A^{opt}}$,
         local computing resources allocation $\boldsymbol{(F^l)^{opt}}$,
         offloading computing resources allocation $\boldsymbol{(F^o)^{opt}}$,
         charging resources allocation $\boldsymbol{P^{opt}}$,
         user operation mode selection indicator $\boldsymbol{\rho^{opt}}$.
    \State Initialize lagrange multipliers $\{{\beta_{j}}\}_{j\in M},\{{\gamma_{j}}\}_{j\in M}$ and the current iteration number $r = 0$.
    \Repeat
    \State Set $r = r + 1$.
    \State For given \{$\boldsymbol{A^{r-1}, (F^o)^{r-1}, P^{r-1}}$\},obtain $\boldsymbol{\rho^r}$ according to (22) and \textbf{Algorithm 1}.
    \State For given \{$\boldsymbol{\rho^r, A^{r-1}, P^{r-1}}$\},obtain $\boldsymbol{(F^o)^r}$ according to (34).
    \State For given \{$\boldsymbol{\rho^r, (F^o)^r, A^{r-1}}$\},obtain $\boldsymbol{P^r}$ according to (47).
    \State For given \{$\boldsymbol{\rho^r, (F^o)^r, P^r}$\},obtain $\boldsymbol{A^r}$ according to (53)-(54) and \textbf{Algorithm 2}.
    \State Computing the value of the objective function(20a).
    \Until{The difference between consecutive values of the objective function(20a) is under $\varepsilon_{2}$.}\\
    \Return $\boldsymbol{A^{opt}, (F^o)^{opt}, P^{opt}, \rho^{opt}}$.
  \end{algorithmic}
\end{algorithm}
\begin{table}[b]	
\centering
\caption{SYSTEM CONFIGURATION ON SIMULATION}\label{table:2}
\resizebox{240pt}{60pt}{\begin{tabular}{c|c}
\toprule
\textbf{Parameter} & \textbf{Value} \\
\midrule
Number of IIoTDs and hovering positions $N$, $M$ & 50, 4 \\
UAV height $H$, area of rigion & 10 $m$, 1000 $m^2$ \\
Bandwidth $B$, Uplink power $p_i$ & 10MHz, 2.83mw \\
Channel gain $g_0$, Noise power ${\sigma ^2}$ & $-30dB$, $-60dB$ \\
Pathloss parameter $\eta_{LoS}$, $\eta_{NLoS}$ & 0.1dB, 21dB [31] \\
Pathloss parameter $a$, $b$ & 4.88, 0.49 [31]\\
Energy conversation efficiency ${\eta _0}$ & 80\% \\
Maximum computation capacity of IIoTDs ${f_{i,max}^{ue}}$ & 1MHz \\
Maximum computation capacity of UAV ${f_{max}^{uav}} $& 3MHz\\
Maximum charging capacity of UAV $p_{max}^{uav}$ & 0.1W \\
Convergence tolerance threshold $\varepsilon_{1}$, $\varepsilon_{2}$ & $10^{-6}, 10^{-10}$ \\
\bottomrule
\end{tabular}}
\end{table}
\section{Performance Analysis}
\subsection{Setting}
\begin{figure}[t]
    \centering
    \subfigure[]{
        \includegraphics[width=0.25\textwidth]{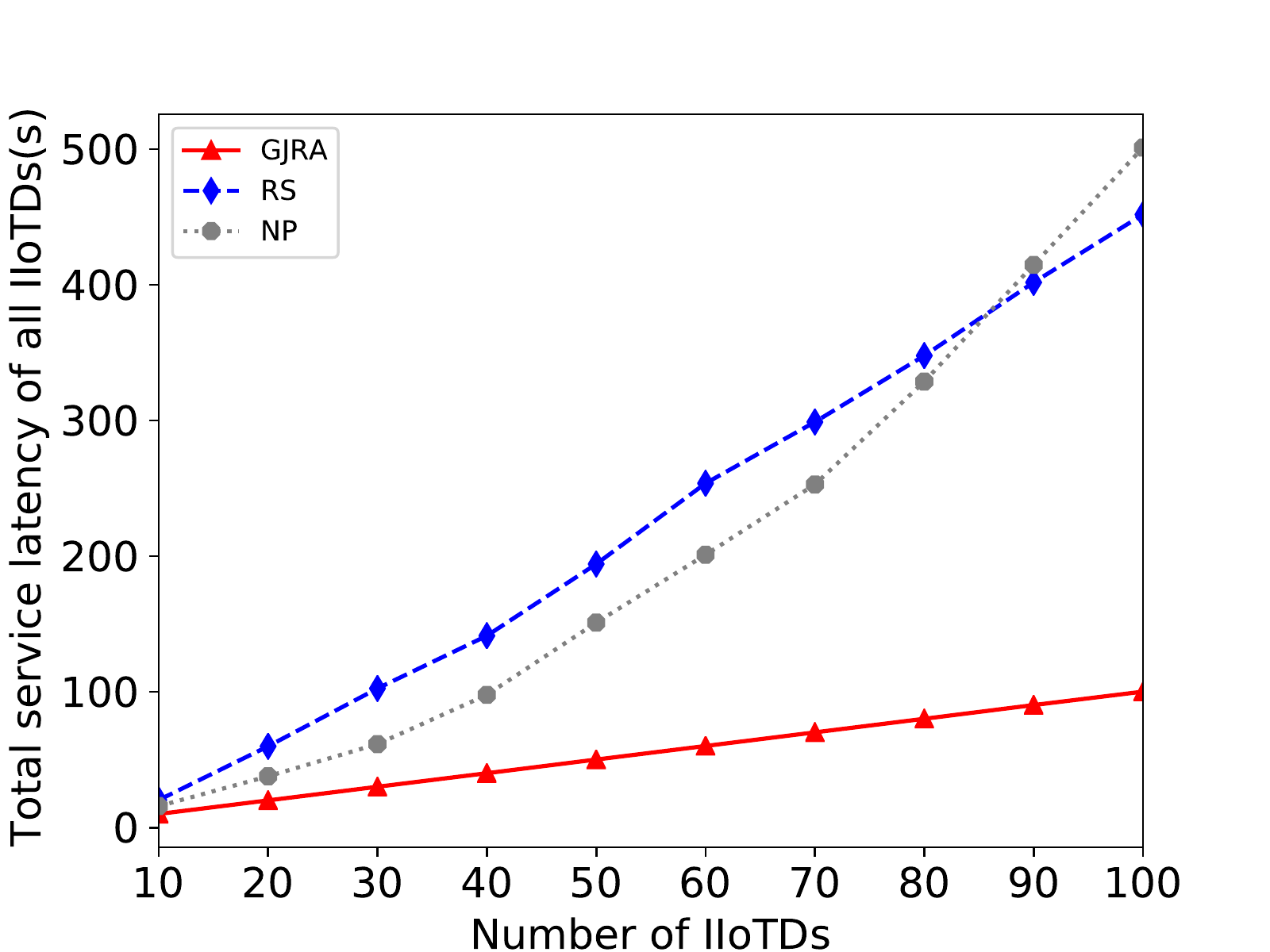}
    }%
    \subfigure[]{
        \includegraphics[width=0.25\textwidth]{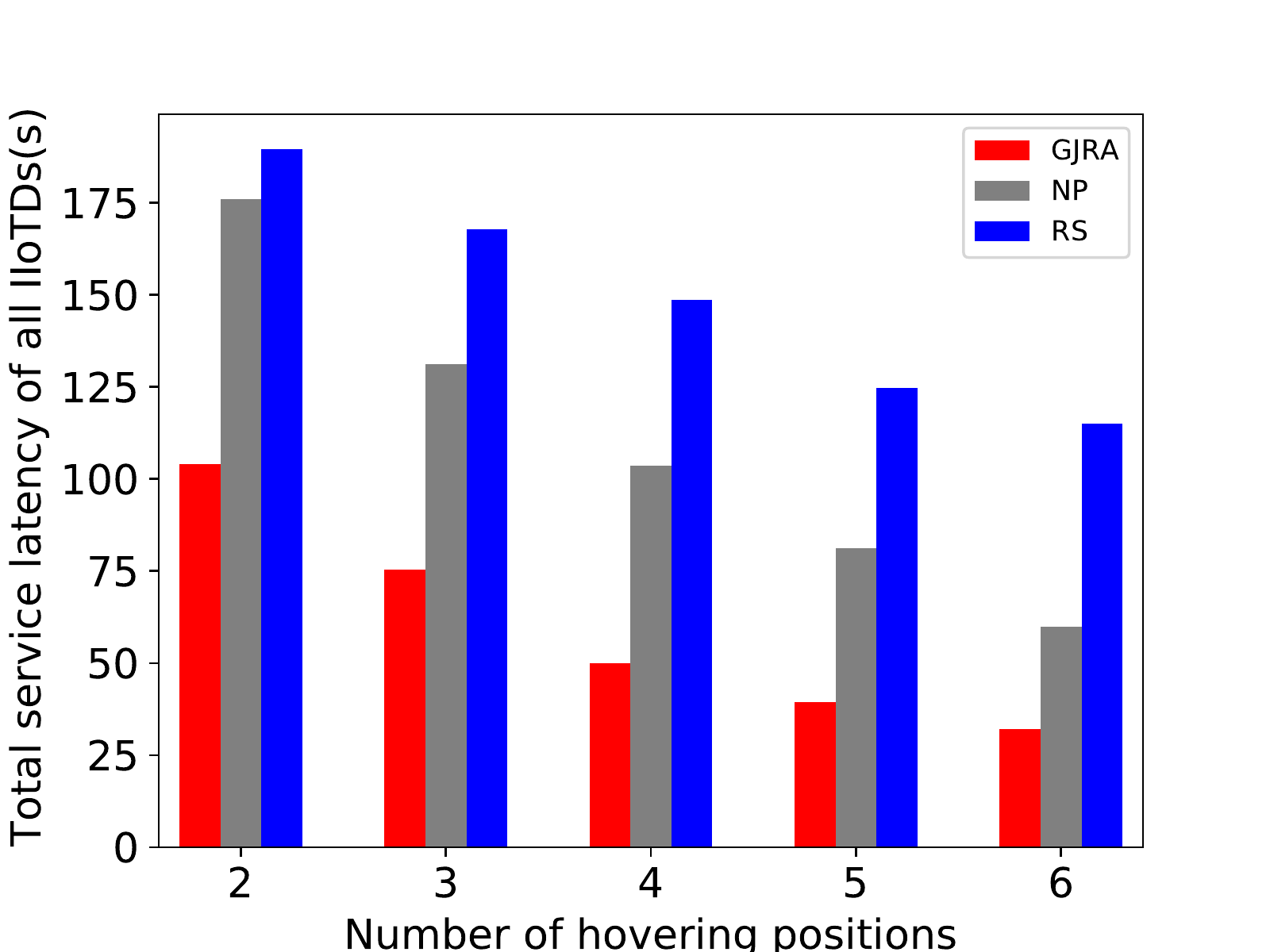}
    }
    \captionsetup{font={small,it}}
    \caption{The total service latency of all IIoTDs versus: (a)the number
of the IIoTDs; (b)the number of the hovering positions}
\end{figure}
In this section, numerical results are presented to evaluate the performance of the proposed \textbf{Algorithm 3} and the benchmark schemes. 
We consider that all the IIoTDs are randomly distributed within a 2D area and the UAV flies over the area hovering \emph{M} random fixed locations. The simulation parameter settings are summarized in \textbf{Table II} unless other specifically notes.

We compare the proposed \textbf{Algorithm 3} with the following intuitive methods as baselines:
\subsubsection{Random Selection Scheme}
The task offloading decisions, the charging resources allocation and the UAV computation resources allocation are optimized while the IIoTDs connection is randomly selected, which is labeled as `RS';
\subsubsection{Nearest Position Scheme}
All of the IIoTDs select the nearest hovering position connecting with the UAV while others variables are optimized, which is labeled as `NP';
\subsubsection{Exhaustive Search Scheme}
The optimal solution of the considering system is obtained after traversing all values within the ranges of all the optimization variables, which is labeled as `EA'.

\subsection{Evaluation}

\setlength{\abovecaptionskip}{-0.1cm}
\begin{figure}[t]
    \centering
    \subfigure[]{
        \includegraphics[width=0.25\textwidth]{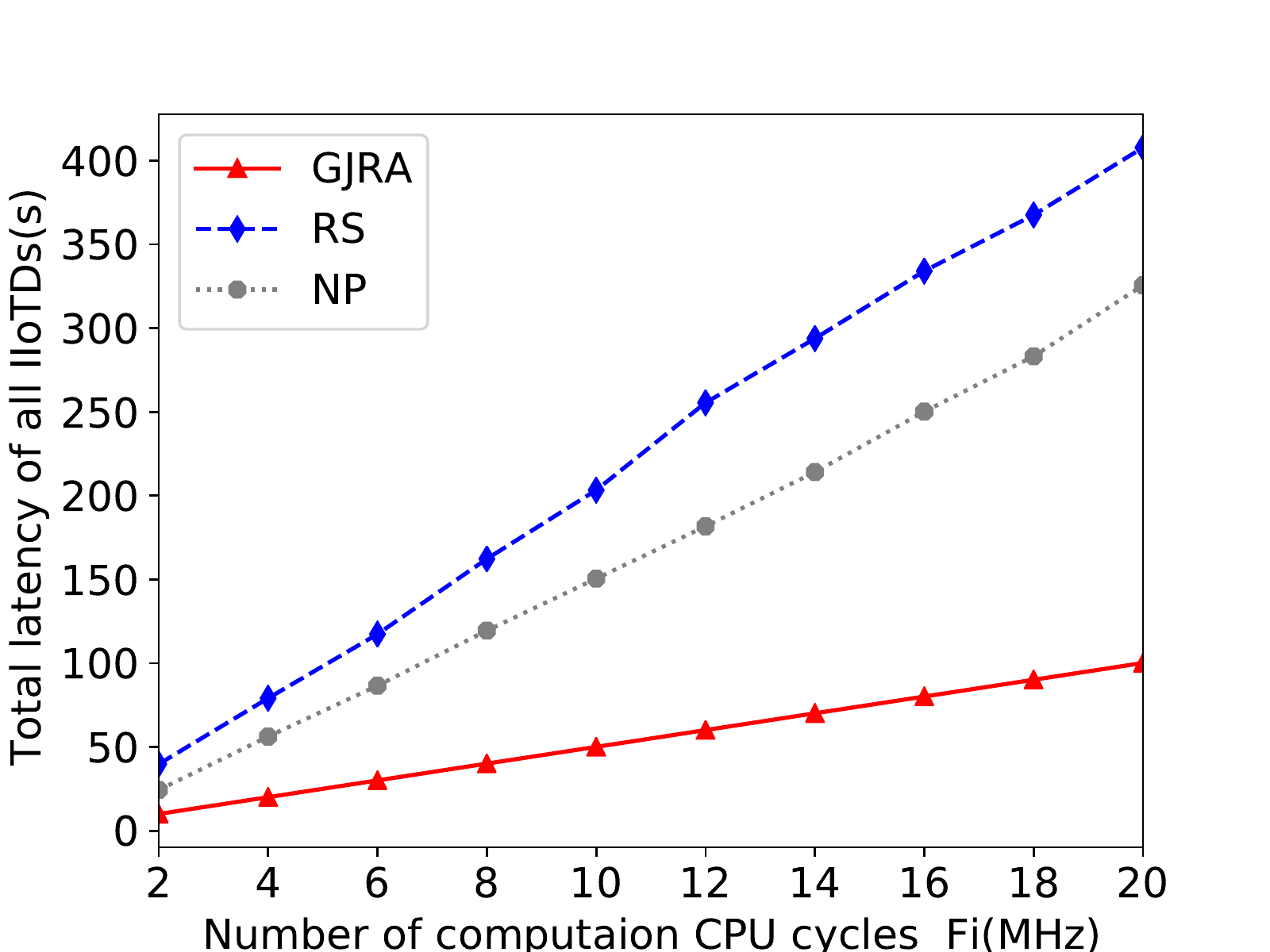}
    }%
    \subfigure[]{
        \includegraphics[width=0.25\textwidth]{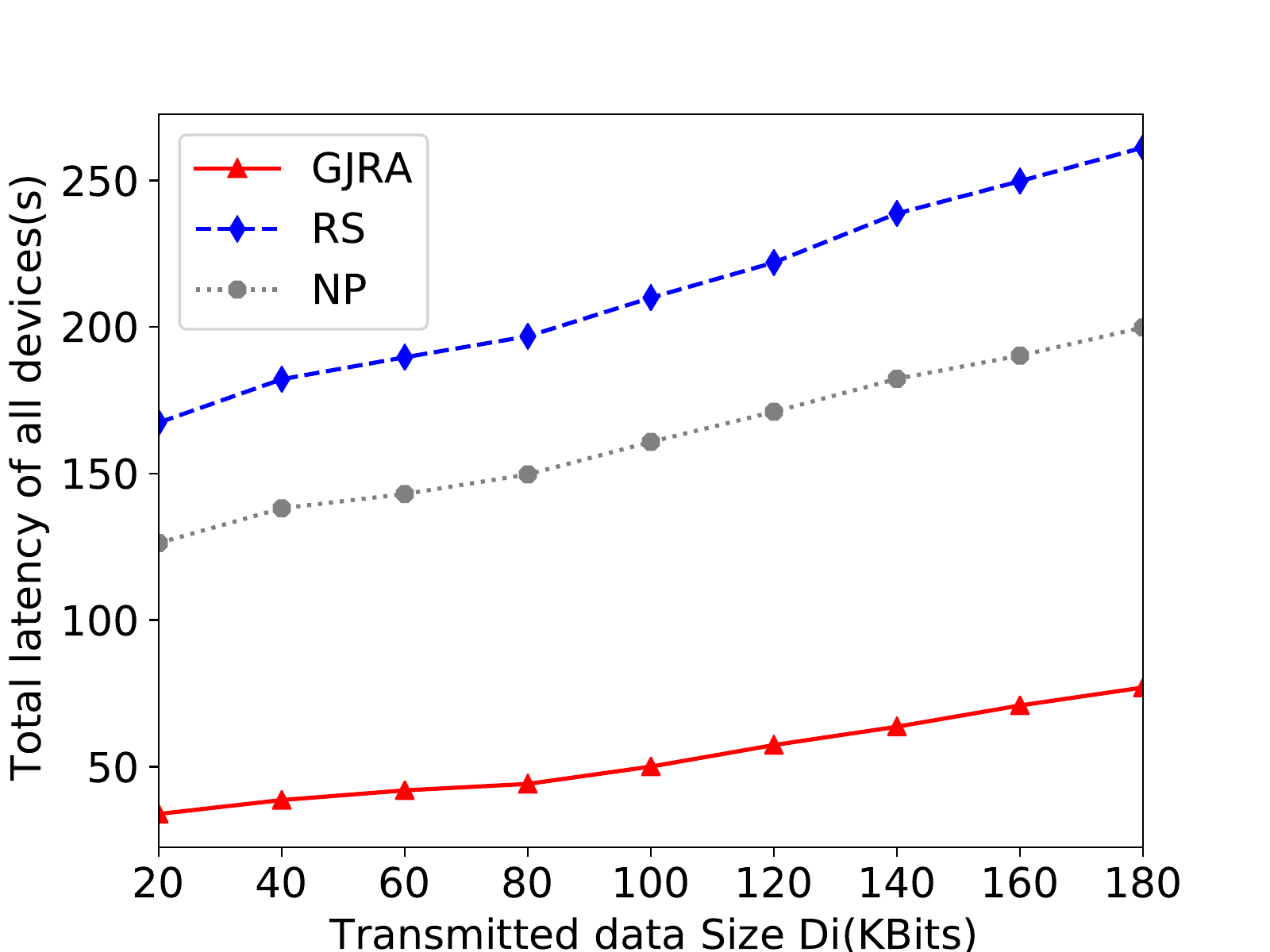}
    }
    \captionsetup{font={small,it}}
    \caption{The total service latency of all IIoTDs versus task settings: (a)the number of task computing CPU cycles; (b)data size}
\end{figure}

\begin{figure}[t]
    \centering
    \subfigure[]{
        \includegraphics[width=0.25\textwidth]{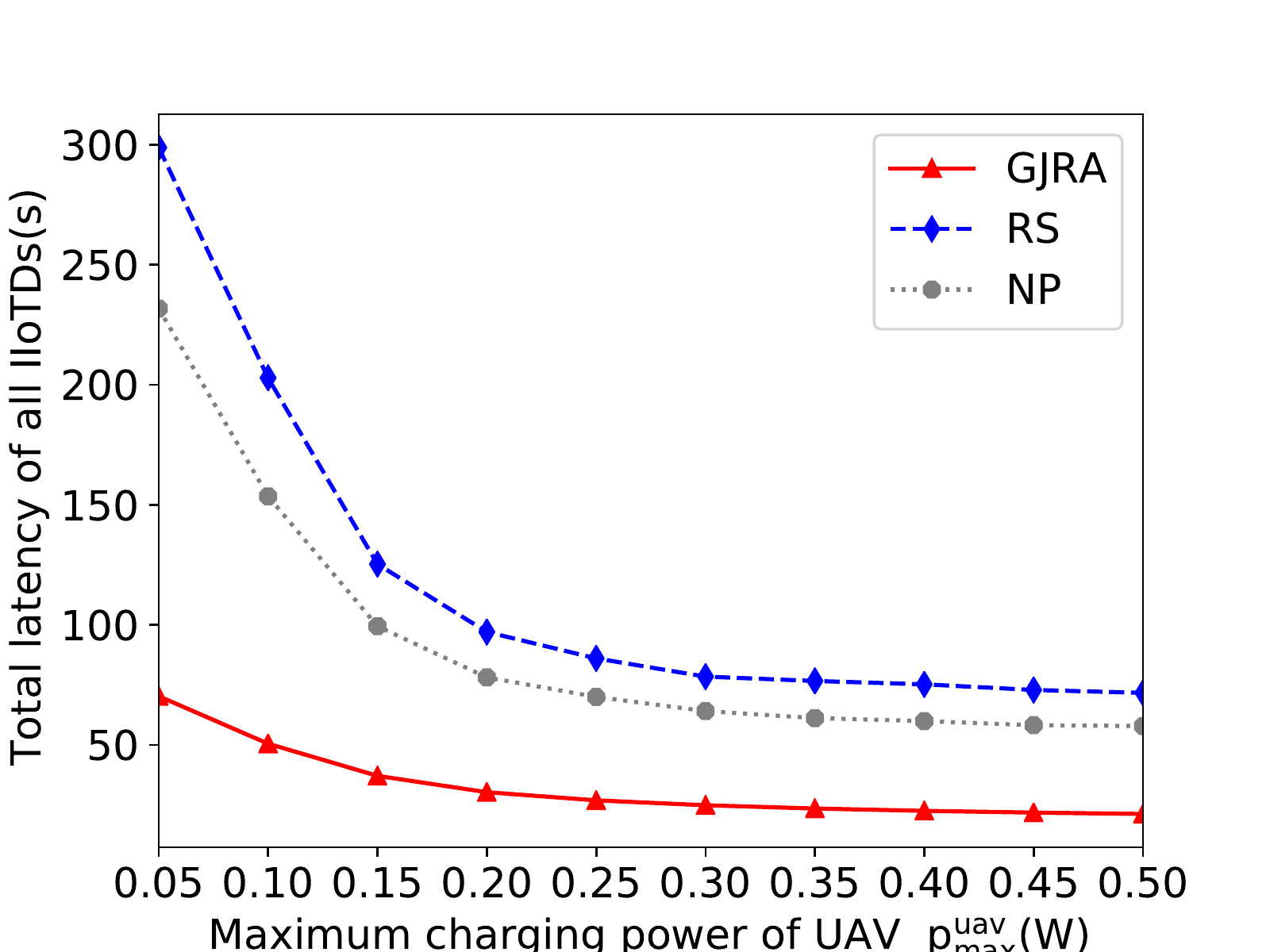}
    }%
    \subfigure[]{
        \includegraphics[width=0.25\textwidth]{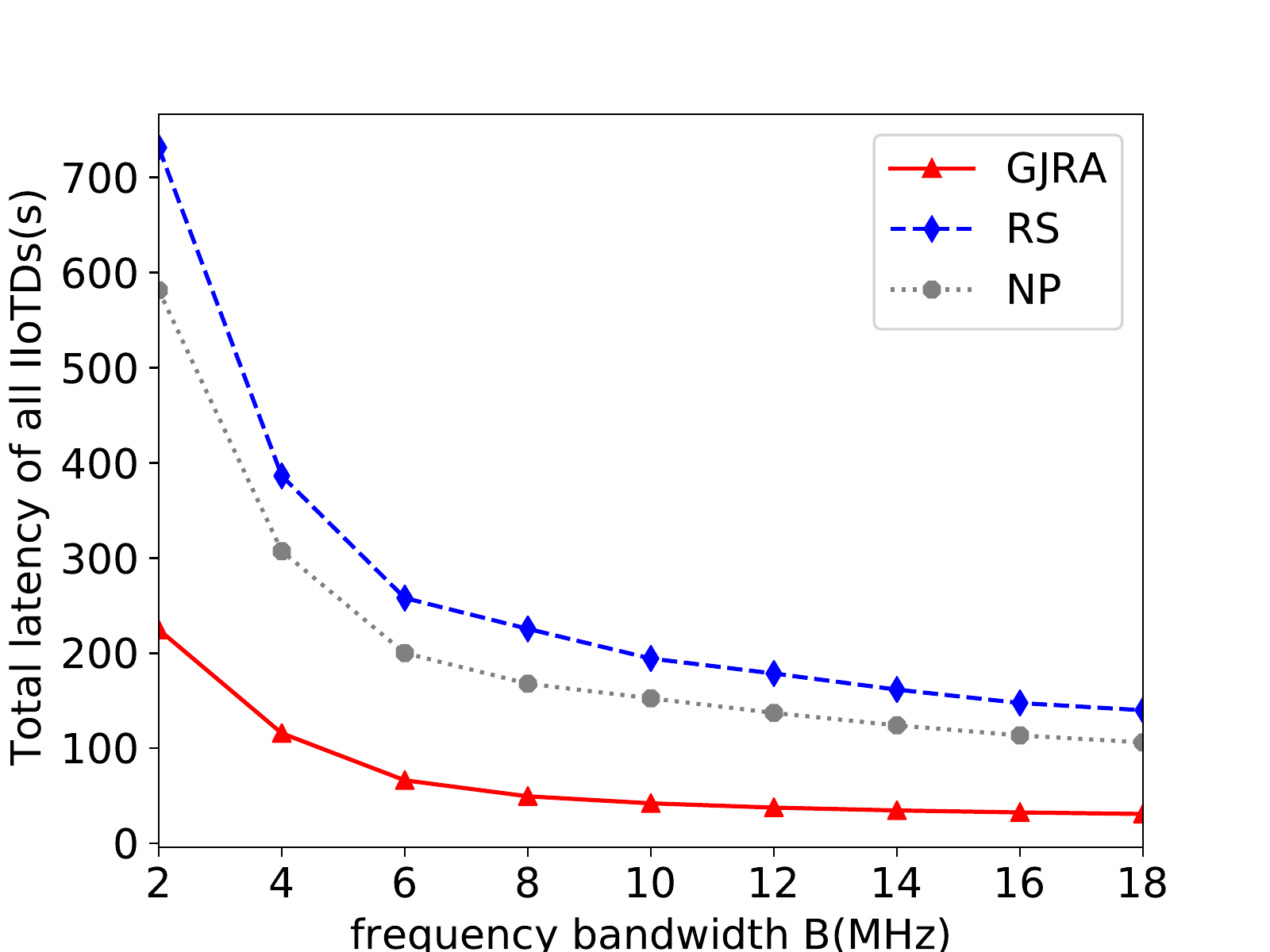}
    }
    \captionsetup{font={small,it}}
    \caption{The total service latency of all IIoTDs versus transmission settings: (a)the maximum charging capacity of UAV; (b)the system frequency bandwidth}
\end{figure}
\begin{figure}[t]
    \centering
    \subfigure[]{
        \includegraphics[width=0.25\textwidth]{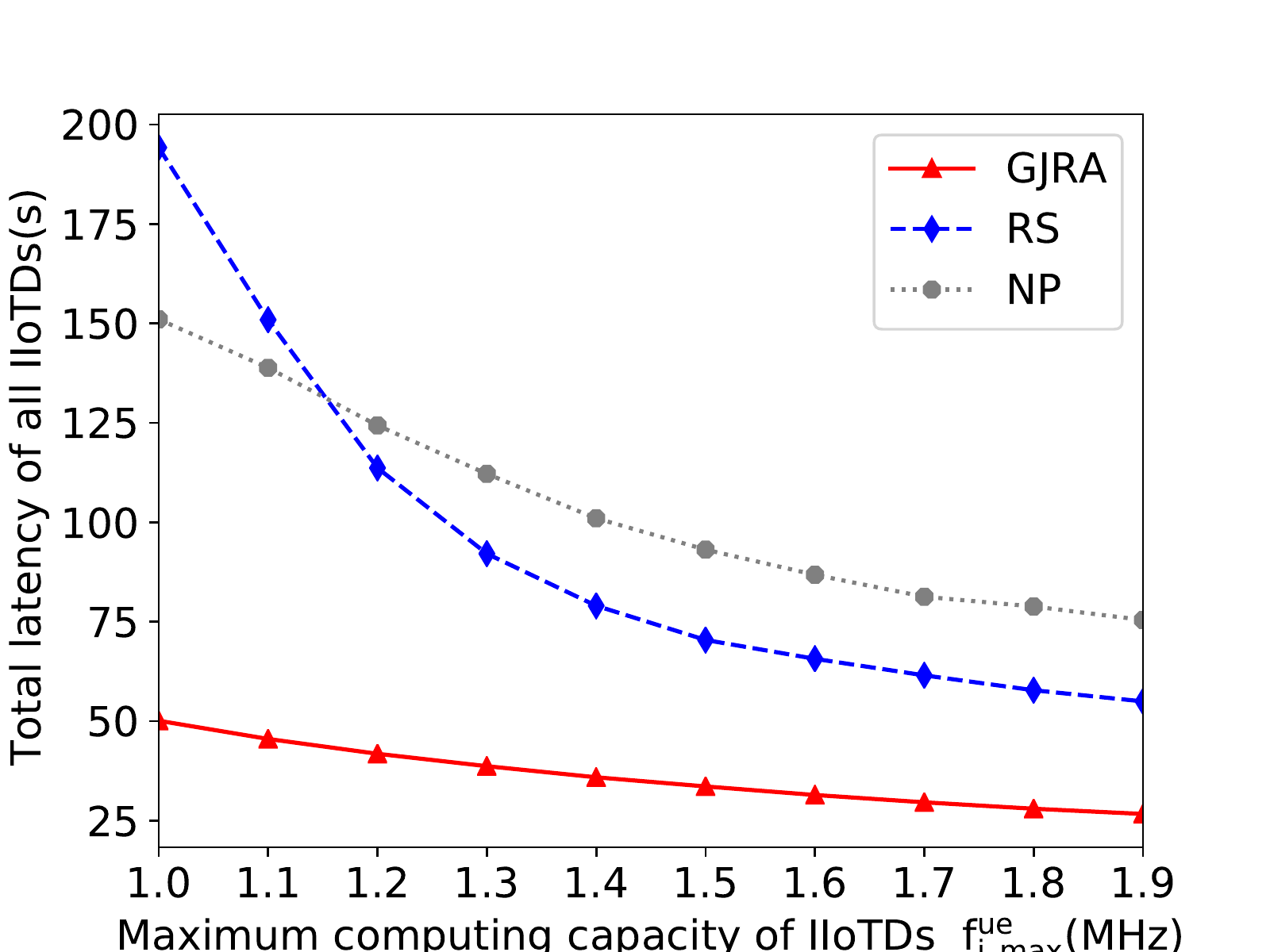}
    }%
    \subfigure[]{
        \includegraphics[width=0.25\textwidth]{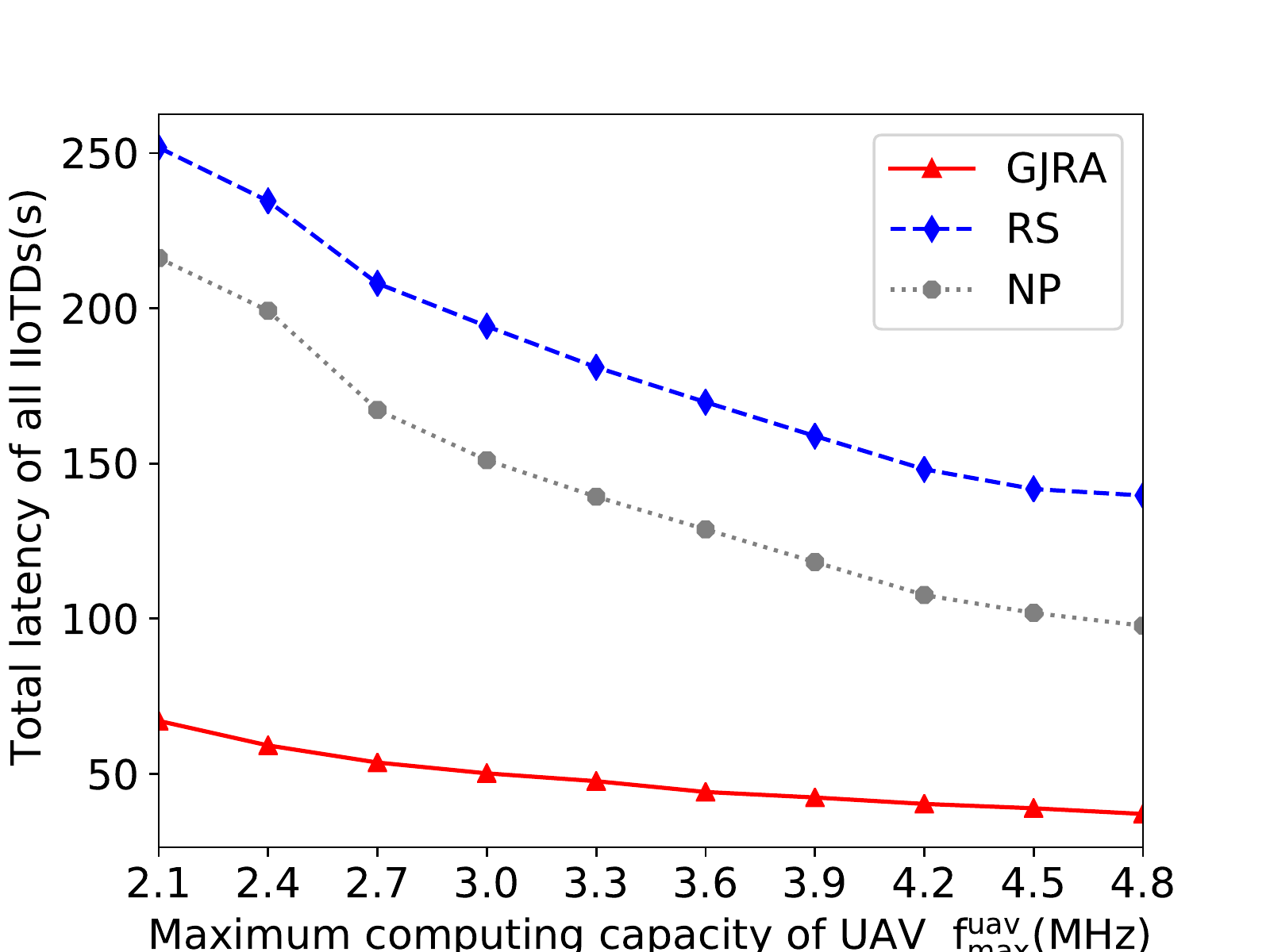}
    }
    \captionsetup{font={small,it}}
    \caption{The total service latency of all IIoTDs versus the maximum computation capability of: (a)IIoTDs; (b)UAV}
\end{figure}
\begin{figure}[t]
    \centering
    \subfigure[]{
        \includegraphics[width=0.25\textwidth]{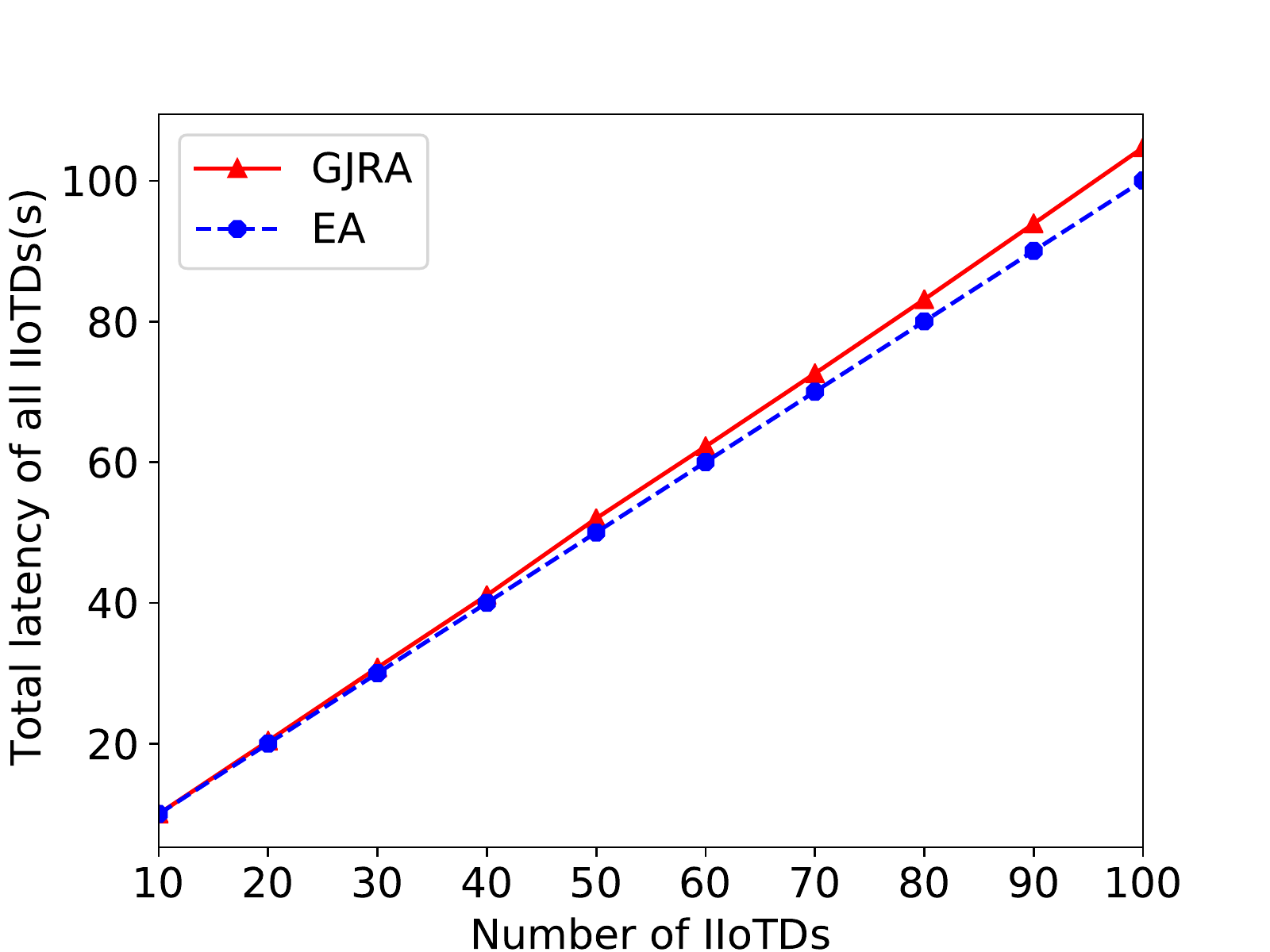}
    }%
    \subfigure[]{
        \includegraphics[width=0.25\textwidth]{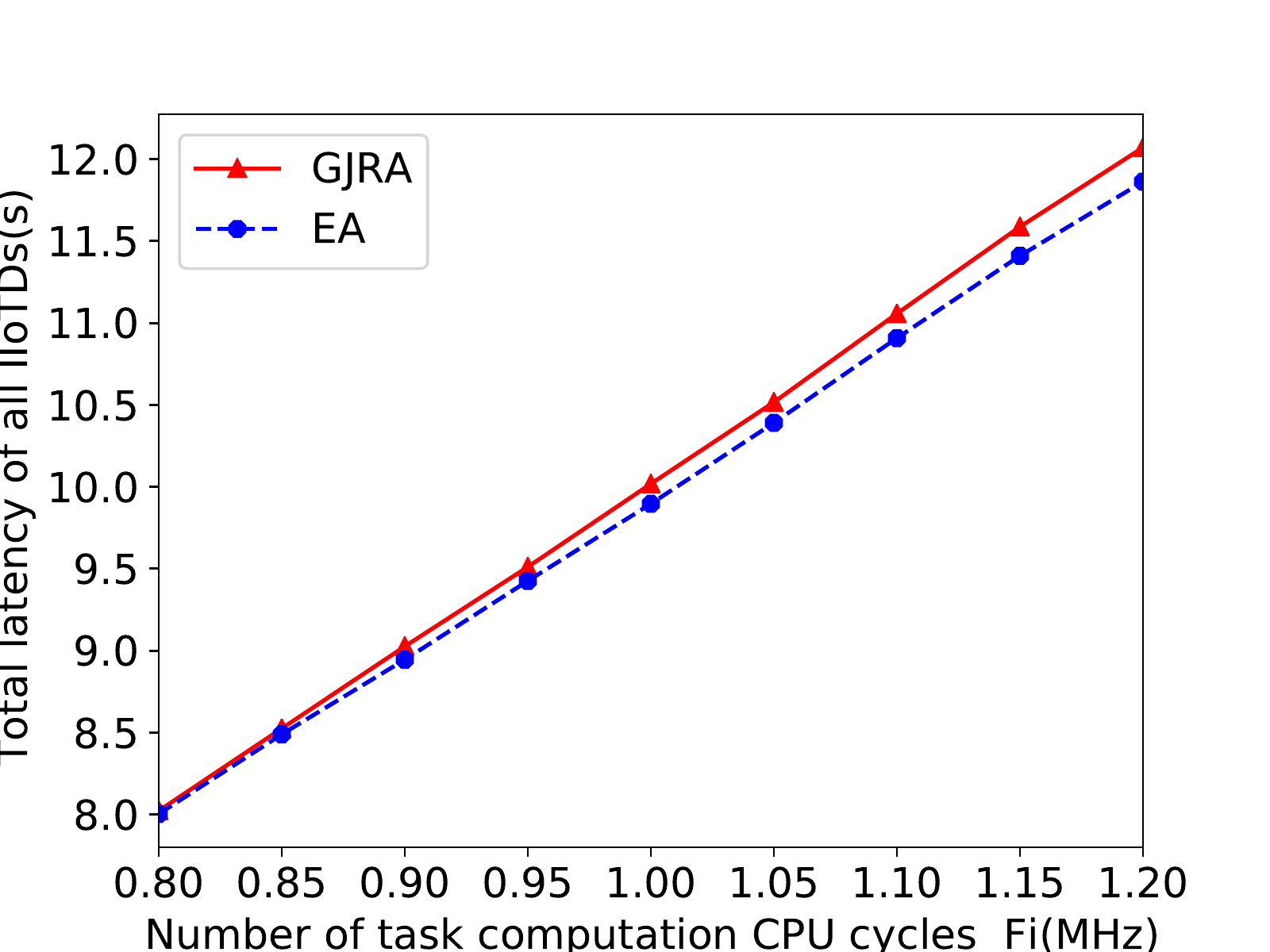}
    }
    \captionsetup{font={small,it}}
    \caption{The gap between our proposed solution and optimal solution with the increasing of: (a)the number of IIoTDs; (b)the number of task computation CPU cycles}
\end{figure}
Fig. 2(a) and Fig. 2(b) illustrates the total service latency of all IIoTDs versus the number
of the IIoTDs \emph{N} and the hovering positions \emph{M}. With the increasing of \emph{N} and \emph{M}, the total service latency of all IIoTDs increases as expected. One can see in Fig. 2(a) that the total service latency of all IIoTDs under GJRA and RS grow at the uniform pace with the increasing of the number of the IIoTDs, while the total service latency under NP grows faster at the same condition, even exceeding RS in the end. This is because when the number of IIoTDs reaches a certain value, certain hovering positions in hot area(i.e., the positions in the center of the IIoTD cluster) will connect to so many IIoTDs at the same time that no matter IIoTDs choose either local execution mode or task offloading mode the service latency will increase dramatically due to the limited resources, while GJRA can effectively avoid the emergence of this problem. One can see in Fig. 2(b) that the gap among our proposed scheme and two benchmark schemes gradually increases with the increasing of the hovering positions. Besides, there is no significant difference between nearest position scheme and random selection scheme with a small amount of UAV hovering position, on account of the similar connection scheme for IIoTD with few alternatives. From the figure, we also can see GJRA outperforms the other two benchmarks.

Fig. 3(a) and Fig. 3(b) shows the total service latency of all IIoTDs versus task settings $F_i$ and $D_i$, respectively. With the increasing both of $F_i$ and $D_i$, the total service latency of all IIoTDs increases, as expected. Besides, by comparing the two figures, we can see that the computing requirement of IIoTDs has a bigger impact to total service latency than communication requirement. Moreover, GJRA outperforms the other two benchmarks and the total service latency of all IIoTDs is significantly reduced by the proposed scheme and algorithms.

Fig. 4(a) and Fig. 4(b) shows the total service latency of all IIoTDs versus transmission settings. With the increasing both of $p_{\max }^{uav}$ and $B$, the energy harvesting time and the data transmission time would decrease correspondingly, which leads the decreases of the total service latency of all IIoTDs, as well as the gaps among three schemes. When $p_{\max }^{uav}$ and $B$ increase to a higher value, the latency of energy harvesting data transmission become too small to impact the total service latency powerfully. Moreover, GJRA outperforms the other two benchmarks.

Fig. 5(a) and Fig. 5(b) illustrate the total service latency of all IIoTDs versus the maximum computing capacity of IIoTDs and UAV. With the increasing both of $f_{i,\max }^{ue}$ and $f^{uav}_{max}$, the total service latency decreases and GJRA outperforms the benchmarks. Notice that with the increasing both of the maximum computing capacity of IIoTDs $f_{i,\max }^{ue}$, more IIoTDs would offload thire tasks to the UAV due to the less competition of UAV computation resources, which results in a slowing down of declines of the total service latency of all IIoTDs under all simulated schemes. Moreover, the total service latency of all IIoTDs under NP decreases slower after the maximum computing capacity of IIoTDs $f_{i,\max }^{ue}$ reaching a lager value, even exceeding RS. This is because that the reduction of competitive pressure on UAV computation resources is smaller than the other two schemes mentioned due to the existence of some hot hovering positions(i.e., the positions in the center of the IIoTD cluster) connecting with more IIoTDs, while GJRA can manage the connection between IIoTDs and UAV intelligently which can avoid the emergence of this problem effectively.

Fig. 6(a) and Fig. 6(b) illustrates the gap between GJRA and optimal solution EA with the increasing of the number of IIoTDs and task computation CPU cycles. In Fig. 6(b), we suppose the total number of IIoTDs in the system is 10. Both of them show the comparison between our proposed solution and the exhaustive search scheme, which can be considered as the optimal solution. One can see that the performance of our algorithm is close to the exhaustive algorithm. However, exhaustive search scheme have to searches all the feasible solution before finding the optimal solution which has the lowest efficiency obviously while GJRA with much less complexity.
\setlength{\abovecaptionskip}{0.1cm}
\begin{figure}[b]
\centering
\includegraphics[scale=0.25]{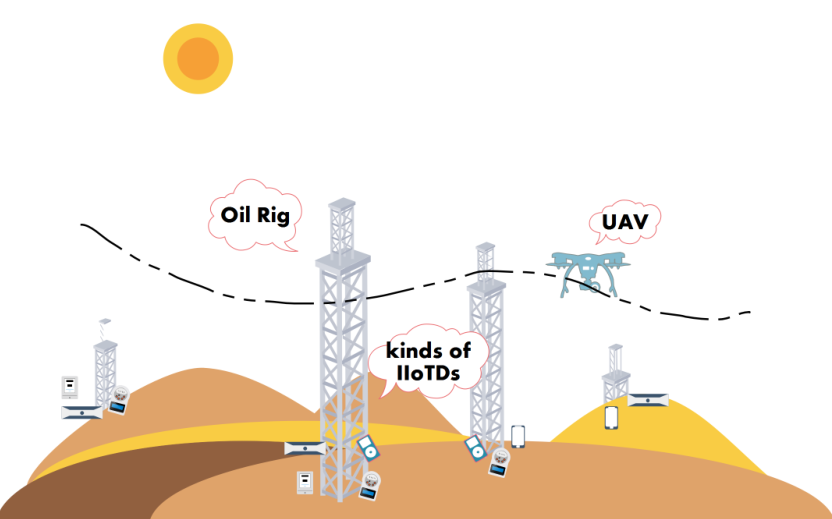}
\captionsetup{font={small,it}}
\caption{The proposed industrial application scene and working process}
\end{figure}
\subsection{Industrial Applications}
In this subsection, we will illustrate the system operation with a practical industrial scenario to show the practicability of our proposed global joint resource allocation scheme for UAV service of PEC in IIoTs.

To take full advantage of the maneuverability of UAVs, we assumed an petroleum exploration project in the desert, as shown in the Fig. 7, where multiple oil rigs distributes in the region as well as kinds of IIoTDs. The working process of the system is shown in the Fig. 8.

\begin{figure}[t]
\centering
\includegraphics[scale=0.30]{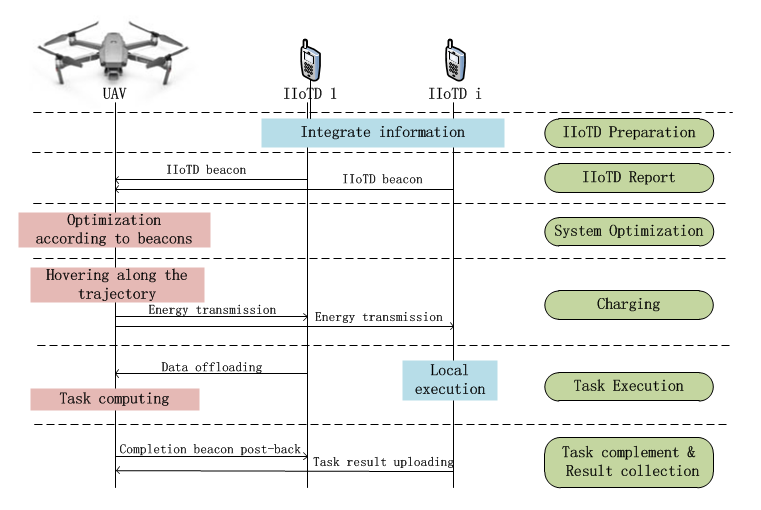}
\captionsetup{font={small,it}}
\caption{The total service latency of all IIoTDs versus the system frequency bandwidth}
\end{figure}
\subsubsection{IIoTD Preparation}
In the IIoTD preparation step, the IIoTDs perform data sensing mode and the computation and communication module is turned off for energy conservation.
\subsubsection{IIoTD Report}
In the IIoTD report step, the UAV flies around the area along the fixed trajectory to perceive the location of all IIoTDs and get beacons from them. The beacon of each IIoTD contains its ID number, task data size and current location.
\subsubsection{System Optimization}
When receiving the beacons of all IIoTDs, the UAV categorizes the IIoTDs into local execution mode and task offloading mode, and pre-allocates the resources by the propose global joint resource allocation scheme.
\subsubsection{Charging}
After preparation and optimization, the UAV flies and hovers with the pre-defined trajectory as a PEC server and a mobile power source. First of all, the UAV transmits power to each connected IIoTD at each hovering position.
\subsubsection{Task Execution}
At the each hovering position, the UAV computes all the data migrated from the IIoTDs which perform task offloading mode, and the IIoTDs perform local execution mode process their task locally.
\subsubsection{Task Complement and Result Collection}
When the tasks are accomplished, the IIoTDs perform local execution mode upload the task results and the UAV post back a completion beacon to the IIoTDs which perform task offloading mode.

\section{Conclusion and Future Work}
In this paper, we presented a global joint resource allocation scheme for UAV service of PEC in IIoTs to prolong the IIoTDs services and enhance the system performance. Specifically, the overall service latency of all IIoTDs was minimized via jointly optimizing the task offloading decisions, charging resources allocation, connection management and UAV computation resources allocation. To solve this MINLP problem, we proposed an two-layer iterative algorithm through solving four sub-problems, developed by BCD method. The performance analysis validated that the total service latency of IIoTDs can be effectively saved by applying our global joint resource allocation scheme. It was shown that the performance achieved by our proposed scheme is superior to the benchmarks. Moreover, the simulation results verified the efficiency of our proposed alternative algorithms and theoretical analysis.

Based on the research of this paper, we will extend our work to multi-UAV deployment and multi-hop PEC network in the future. Besides, the trajectory design, tradeoff between energy consumption and latency for UAV service of PEC in IIoTs system are also worth of further investigation. Moreover, the framework of Cloud-Edge-Device collaboration in cellular networks cloud be the directions for future work.


%
%

\ifCLASSOPTIONcaptionsoff
  \newpage
\fi

\begin{IEEEbiography}[{\includegraphics[width=1in,height=1.25in,clip,keepaspectratio]{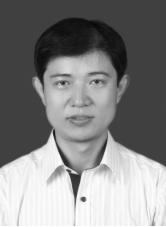}}]{Jin Wang}
 (SM’18) received the B.S. and M.S. degrees from the Nanjing University of Posts and Telecommunications, China, in 2002 and 2005, respectively, and the Ph.D. degree from Kyung Hee University Korea, in 2010. He is currently a Professor with the School of Computer and
Communication Engineering, Changsha University of Science and Technology. His research
interests mainly include wireless communications and networking, performance evaluation, and
optimization. He is a member of the ACM.
\end{IEEEbiography}

\begin{IEEEbiography}[{\includegraphics[width=1in,height=1.25in,clip,keepaspectratio]{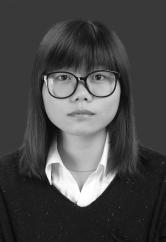}}]{Caiyan Jin}
received the B.E. degree in Internet of Things Engineering from the school of information science and technology, Southwest Jiaotong University, Chengdu, China, in 2018. Her current reserch interests include wireless UAV communication and mobile edge computing.
\end{IEEEbiography}

\begin{IEEEbiography}[{\includegraphics[width=1in,height=1.25in,clip,keepaspectratio]{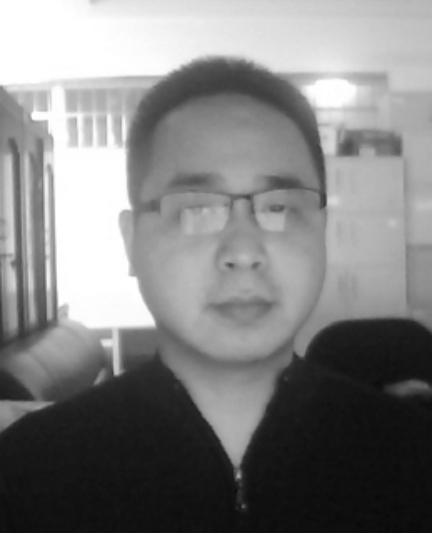}}]{Qiang Tang}
received the B.E., M.S., and Ph.D. degrees in control science and engineering from the Huazhong University of Science and Technology, Wuhan, China, in 2005, 2007, and 2010, respectively. He is an academic visitor sponsored by CSC in University of Essex during 2016-2017. He is currently a Lecturer with the School of Computer and Communication Engineering, Changsha University of Science and Technology, Changsha, China. His research interests include wireless networks, mobile edge computing, and smart grid.
\end{IEEEbiography}

\begin{IEEEbiography}[{\includegraphics[width=1in,height=1.25in,clip,keepaspectratio]{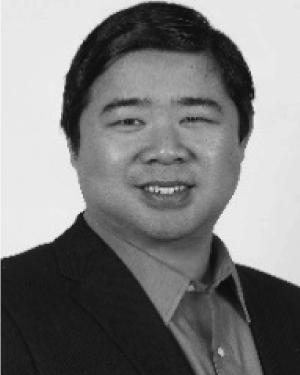}}]{Naixue Xiong}
received the Ph.D. degree in sensor system engineering from Wuhan University and the Ph.D. degree in dependable sensor networks from the Japan Advanced Institute of Science and Technology., Before he attended Tianjin University, he worked in Northeastern State University, Georgia State University, Wentworth Technology Institution, and Colorado Technical University (Full Professor about five years) about ten years. He is currently a Professor with the College of Intelligence and Computing, Tianjin University, China. He has published over 300 international journal articles and over 100 international conference papers. Some of his works were published in IEEE JSAC, IEEE or ACM Transactions, ACM Sigcomm workshop, IEEE INFOCOM, ICDCS, and IPDPS. His research interests include cloud computing, security and dependability, parallel and distributed computing, networks, and optimization theory.,Dr. Xiong is a Senior Member of the IEEE Computer Society. He has received the Best Paper Award in the 10th IEEE International Conference on High Performance Computing and Communications (HPCC-08) and the Best student Paper Award in the 28th North American Fuzzy Information Processing Society Annual Conference (NAFIPS2009). He is also the Chair of the Trusted Cloud Computing Task Force, the IEEE Computational Intelligence.
\end{IEEEbiography}


%




\end{document}